\documentclass[twocolumn]{autart}    

\usepackage{graphicx}          
\usepackage{epsfig} 
\usepackage{amsmath} 
\usepackage{amssymb}  
\usepackage{hyperref}
\usepackage{mathtools}
\usepackage{subcaption}
\usepackage{cancel}
\usepackage{xcolor}
\usepackage{bm}
\usepackage{makecell}
\usepackage{tabulary,capt-of}
\usepackage{enumerate}
\usepackage{systeme}
\usepackage{algorithm}
\usepackage{algorithmic}

\usepackage{amsfonts}
\usepackage{enumitem}
\newcommand{\fref}[1]{Fig.~\ref{#1}}
\newcommand{\sref}[1]{Section~\ref{#1}}
\newcommand{\aref}[1]{Assump.~\ref{#1}}
\newcommand{\lref}[1]{Lemma~\ref{#1}}
\newcommand{\tref}[1]{Theorem~\ref{#1}}
\newcommand{\cref}[1]{Corollary~\ref{#1}}
\newcommand{\dref}[1]{Definition~\ref{#1}}
\newcommand{\pref}[1]{Proposition~\ref{#1}}

\newcommand{\projection}[1]{\mathsf{P}_{#1}}
\newcommand{\norm}[2]{\lVert#1\rVert_{#2}}

\newcommand{\inner}[3]{\langle#1, #2\rangle_{#3}}
\newcommand*{\myqed}{\null\nobreak\hfill\ensuremath{\square}}%

\DeclarePairedDelimiter{\Span}{\text{span}\{}{ \} }

\DeclarePairedDelimiter{\Diag}{ \text{diag}\{ }{ \} }

\DeclareMathOperator{\Adj}{adj}

\DeclareMathOperator*{\argmin}{argmin}

\begin{document}

\begin{frontmatter}


\thanks{The authors are with the Research Center for Systems and Technologies (SYSTEC-ARISE), Faculdade de Engenharia, Universidade do Porto, Portugal
{\tt\small \{matheus.reis, pedro.aguiar\}@fe.up.pt}. This work was funded by the Portuguese Foundation for Science and Technology (FCT) through grant 2020.06795.BD, DOI: 10.54499/2020.06795.BD 
and the Associate Laboratory ARISE – Advanced Production and Intelligent Systems (LA/P/0112/2020, DOI: 10.54499/LA/P/0112/2020), and the SYSTEC - Research Center for Systems and Technologies, funded by FCT/MECI through national funds.}

\title{On the Stability of Undesirable Equilibria in the Quadratic Program Framework for Safety-Critical Control}

\author{Matheus F. Reis},    
\author{A. Pedro Aguiar}    

\address{Research Center for Systems and Technologies (SYSTEC), ARISE, \\
Faculty of Engineering, University of Porto, Portugal}  


\begin{abstract}                          
Control Lyapunov functions (CLFs) and Control Barrier Functions (CBFs) have been used to develop provably safe controllers by means of quadratic programs (QPs). This framework guarantees safety in the form of trajectory invariance with respect to a given set, but it can introduce undesirable equilibrium points to the closed-loop system, which can be asymptotically stable.
In this work, we present a detailed study of the formation and stability of equilibrium points for { nonlinear, control-affine systems} with the CLF-CBF-QP framework with multiple CBFs.
In particular, we show that the stability of undesirable equilibrium points is dependent on the CLF and CBF geometrical properties.
We introduce the concept of CLF-CBF compatibility, regarding a CLF-CBF pair inducing no stable equilibrium points other than the CLF global minimum on the corresponding closed-loop dynamics.
{ Considering LTI and drift-less full-rank systems}, sufficient conditions for CLF-CBF compatibility with quadratic CLF and CBFs are derived, and we propose a novel control strategy to induce smooth changes in the CLF geometry at certain regions of the state space, aiming to satisfy the CLF-CBF compatibility conditions. { The strategy can be used with multiple safety objectives while avoiding the convergence of trajectories towards undesired equilibrium points.}
Numerical examples and simulations illustrate the proposed method and its applicability.
\vspace{-5mm}
\end{abstract}

\end{frontmatter}

\section{Introduction}

The engineering of {\it safety-critical systems} is a fruitful and rich topic receiving a growing amount of attention nowadays.
Safety-critical systems are of crucial importance for many industrial sectors and production lines, where the stability of feedback-controlled systems is just so important as their capacity to provide safe behaviour under a wide variety of operational circumstances.
Furthermore, safety is also a mandatory property for systems with high levels of interoperability, cooperation, or coordination with humans.

The notion of {\it safety} was first introduced in 1977 in the context of program correctness by \cite{Lamport1977} and later formalized in \cite{Alpern1985}, which also introduced the concept of {\it liveness}. Intuitively, one can describe these two contrasting system properties as: (i) the requirement of avoiding undesired situations while (ii) guaranteeing the eventual achievement of a desired configuration, respectively.
As pointed out by \cite{Ames2019}, in the context of control systems, liveness can be identified as an {\it asymptotic stability} requirement with respect to a certain set of desired or objective states, while safety can be defined as the {\it invariance} of the system trajectories to some set, defined as the set of {\it safe} states.
While the design of asymptotically stabilizing controllers has been extensively studied in control Lyapunov theory \cite{Khalil2002}, the design of controllers capable of guaranteeing 
safety has been the subject of study in the topic of Control Barrier Functions (CBFs) \cite{Wieland2007}.
%
%
%
%
More recently, \cite{Ames2014} introduced the idea of unifying CBFs with Control Lyapunov Functions (CLFs) through the use of quadratic programs (QPs), combining safety and stabilization requirements in a single control framework.
%

However, the study of controllers combining the two desirable properties of {\it stability} and {\it safety} is still in early stages. 
In \cite{Reis_LCSS}, it is shown that the QP-based framework proposed by \cite{Ames2014} can introduce undesirable equilibrium points other than the CLF minimum into the closed-loop system. The fact that some of these undesirable equilibrium points can be asymptotically stable and can be arbitrarily close to the set of unsafe states is an important practical limitation of the framework, since it could result in system deadlocks and expose the system to close-to-failure situations, forcing the designer to opt for highly conservative safety margins when designing the system safety specifications.
%
%
In \cite{Notomista2022}, a CBF-based controller was proposed in which safety is ensured with respect to multiple non-convex unsafe regions, and undesirable stable equilibria are practically avoided. Still, the method is dependent on the computation of a nonlinear ``convexification'' function for the unsafe sets, which is dependent on the barrier geometry and could be computationally difficult to solve. Furthermore, it is unclear how these results can be generalized to the CLF-CBF framework.
In \cite{Wang2018_multirobot}, the problem of deadlocks in the QP-based formulation for safety-critical systems was addressed for the safety-filter CBF-QP-based controller 
as proposed in \cite{Ames2019}. In this context, deadlocks are caused by a conflict between the stabilization objectives of the nominal controller and the safety barriers, and were 
managed by introducing a consistent perturbation into the QP constraints. Although efficient for solving some types of deadlocks, this proposed method modifies the safety constraints and
allows for the possibility of leaving the safe set if the deadlock situation happens to occur on the boundary, which can lead to unsafe behavior.
Considering the CLF-CBF framework, \cite{tan2024undesired} has proposed a modified CLF-CBF-based QP controller in which {\it interior} equilibrium points and certain {\it boundary equilibrium points} 
satisfying a certain condition do not exist for the resulting closed-loop system. However, boundary equilibrium points could still occur in general. 
%

The contributions of the present work are:
\begin{enumerate}[label=\roman*)]
\item We present an analysis of the existence and stability conditions for equilibrium points occurring in the CLF-CBF-QP-based framework \cite{Ames2014} for nonlinear, control affine systems (Sections \ref{sec:existence}-\ref{sec:stability}).
\item We introduce the concept of {\it CLF-compatibility}, denoting the property of a CLF that, for a given system dynamics and set of CBFs modeling the safety requirements for the task, guarantees that the CLF-CBF-QP controller does not introduce any stable equilibria other than the CLF global minimum  (\sref{sec:clf_compatibility}).
\item We derive necessary and sufficient conditions for CLF compatibility for quadratic CLF/CBFs and for the following classes of systems: (i) control linear systems with full-rank input matrix and (ii) linear time-invariant (LTI) systems (\sref{sec:prob_formulation}).
\item For quadratic CLF/CBFs and for the types of systems (i) and (ii), we propose a method for computing a corresponding compatible CLF from a non-compatible one, and a CLF-CBF-QP controller that adaptively modifies the CLF geometry to achieve CLF-compatibility using the found compatible CLF (\sref{sec:proposed_controller}).
\end{enumerate}

\section{Preliminaries}
\label{sec:qp-based-control-analysis}

\subsection{Notation}
\label{sec:notation}
The fields of real and complex numbers are $\mathbb{R}$ and $\mathbb{C}$, respectively.
Given a matrix $A \in \mathbb{R}^{n \times m}$, 
$[A]_{ij} \in \mathbb{R}$ denotes its $i$-th row, $j$-th column component and
$[A]_{k} \in \mathbb{R}^n$ denotes its $k$-th column.
The group of real symmetric matrices is $\mathcal{S}^n \subset \mathbb{R}^{n \times n}$.
The determinant of a square matrix $A \in \mathbb{R}^{n \times n}$ is $\det A$ or $|A|$, its Frobenius norm is $\norm{A}{\mathcal{F}}$, and 
its adjoint matrix is $\Adj{A} \in \mathbb{R}^{n \times n}$, where $A \, \Adj{A} = |A| I_n$, where 
$I_n \in \mathbb{R}^{n \times n}$ is the $n \times n$ identity matrix.
Given a vector $v \in \mathbb{R}^{n}$, $[v]_k \in \mathbb{R}$ is its $k$-th component.
A scalar-valued function $f: \mathbb{R}^n \rightarrow \mathbb{R}$ is said to be of (differentiability) class $C^k$ if all of its $k$-th order partial derivatives exist and are continuous.
Consider the class $C^2$ function $f: \mathbb{R}^n \rightarrow \mathbb{R}$:
(i) its {\it gradient} is defined as the vector-valued function $\nabla f: \mathbb{R}^n \rightarrow \mathbb{R}^n$ such that
$[\nabla f(x)]_k = \frac{\partial f(x)}{\partial x_k} = \partial_k f(x)$, where $\partial_k$ denotes partial differentiation with respect to the $k$-th component of the function input,
(ii) its {\it Hessian} matrix is defined as the matrix-valued function $H_f: \mathbb{R}^n \rightarrow \mathcal{S}^n$ such that $[H_f(x)]_{ij} = \frac{\partial^2 f(x)}{\partial x_i \partial x_j}$.
$L_g f$ is the Lie derivative of $f$ along $g: \mathbb{R}^{n} \rightarrow \mathbb{R}^{n \times m}$, that is, $L_g f = \nabla f^\mathsf{T} g \in \mathbb{R}^{m}$.
%
%
The inner product between two vectors $u, v \in \mathbb{R}^n$ induced by a positive semidefinite matrix $G = G^\mathsf{T} \ge 0$ is given by $\inner{u}{v}{G} = u^\mathsf{T} G \, v$.
This inner product induces a norm $\norm{v}{G}^2 = \inner{v}{v}{G} = v^\mathsf{T} G \, v$ over $\mathbb{R}^n$. The standard inner product is then $\inner{u}{v}{} = \inner{u}{v}{I_n}$, with standard Euclidean norm $\norm{v}{}^2 = \norm{v}{I_n}^2$.
The orthogonal complement of a subspace $\mathcal{W}$ is denoted by $\mathcal{W}^\perp$, with the notion of orthogonality dependent upon the considered inner product $\inner{\cdot}{\cdot}{G}$.
The set 
$\Span{v_1, \cdots, v_p}$ is the set of all linear combinations of vectors from $\{v_1, \cdots, v_p \} \subset \mathbb{R}^n$. 
%
The positive semi-definite cone of symmetric matrices is $\mathcal{S}^n_+ \subset \mathcal{S}^n$.
%
%
The null space and spectrum of $A \in \mathbb{R}^{n \times n}$ are given by $\mathcal{N}(A) \subset \mathbb{R}^n$, $\sigma(A) = \{ \lambda \in \mathbb{C} \,|\, \det(\lambda I_n - A) = 0 \} \subset \mathbb{C}$, respectively.

\subsection{Polynomial Matrices}

A real polynomial matrix of degree $d \in \mathbb{N}$ is a matrix-valued polynomial function $P: \mathbb{R} \rightarrow \mathbb{R}^{n \times n}$ defined as $P(\lambda) = \sum^d_{i=0} \lambda^i A_i$, where $A_0, A_1, \cdots, A_p \in \mathbb{R}^{n \times n}$.
The Generalized Eigenvalue Problem (GEP) for $P$ consists of finding all the pairs 
$(\lambda, q) \in \mathbb{C} \times \mathbb{R}^n, q \neq 0$ satisfying $P(\lambda) q = 0$.
The spectrum of a polynomial matrix $P$ is defined as the set of generalized eigenvalues $\lambda \in \mathbb{C}$ solving the GEP.
This set is denoted by $\sigma(P) = \{ \lambda \in \mathbb{C} : \det P(\lambda) = 0 \} \subseteq \mathbb{C}$, and
$P$ is said to be a {\it regular} pencil if $\sigma(P) \subset \mathbb{C}$, that is, if $\det P(\lambda)$ does not vanish identically.
If $P$ is regular, $\det P(\lambda)$ is a polynomial of maximum degree $n d$. Therefore, $\sigma(P)$ has a maximum of $n d$ generalized eigenvalues in $\mathbb{C}$.
A polynomial matrix $P$ may have intervals of definiteness, that is, a collection of intervals $\{ \mathcal{I}_1, \mathcal{I}_2, \cdots \}$, such that
$P(\lambda) \gtrless 0 \,\forall \lambda \in \mathcal{I}_i \subset \mathbb{R}$ \cite{Higham2009}.
A real linear matrix pencil (LMP) or a real, affine polynomial matrix is a real polynomial matrix of degree $1$, that is, a function
$P: \mathbb{R} \rightarrow \mathbb{R}^{n \times n}$ defined by a matrix pair $(A_0,A_1) \in \mathbb{R}^{n \times n} \times \mathbb{R}^{n \times n}$ as $P(\lambda) = A_0 + \lambda A_1$ \cite{Ikramov1993}.
%
%
In particular, if $P$ is regular, then it has exactly $n$ generalized eigenvalues.
Any LMP defined by the matrix pair $(A_0,A_1)$ can be decomposed by the generalized Schur decomposition: for any matrix pair $(A_0,A_1) \in \mathbb{R}^{n \times n} \times \mathbb{R}^{n \times n}$, 
there exist orthogonal matrices $Q, Z$ such that $Q (A_0,A_1) Z^\mathsf{T} = (S,T)$, where $S \in \mathbb{R}^{n \times n}$ is block upper triangular and $T \in \mathbb{R}^{n \times n}$ is upper triangular.
The GEP for an LMP can be efficiently solved by the QZ algorithm, which computes the generalized Schur decomposition in a numerically stable manner \cite{Wilkinson1979}.
The spectrum of $P$ can then be computed from the block diagonal elements of $S$ and $T$.

\subsection{CLF-CBF-Based Safety Critical Control}
Consider the nonlinear control affine system
\begin{align}
\dot{x} = f(x) + g(x) u
\label{eq:affine_nonlinear}
\end{align}
where $x \in \mathbb{R}^n$ is the system state and $u \in \mathbb{R}^m$ is the control input. Vector fields $f: \mathbb{R}^n \rightarrow \mathbb{R}^n$, $g: \mathbb{R}^n \rightarrow \mathbb{R}^{n \times m}$ are locally Lipschitz.
%
%

%

\vspace{-2mm}
\begin{defn}[CLFs]
A positive definite function $V$ is a {\it control Lyapunov function} (CLF) for system \eqref{eq:affine_nonlinear} if it satisfies:
%
\begin{align}
\inf_{u\in \mathbb{R}^m} \left[L_f V(x) + \inner{L_g V(x)}{u}{} \right] \le -\gamma(V(x)) \nonumber
\end{align}
where $\gamma: \mathbb{R}_{\ge 0} \rightarrow \mathbb{R}_{\ge 0}$ is a class $\mathcal{K}$ function \cite{Khalil2002}.
\end{defn}
\vspace{-2mm}
This definition implies that there exists a set of stabilizing controls that makes the CLF strictly decreasing everywhere outside its global minimum $x_0 \in \mathbb{R}^n$. 
%
%
%

\begin{defn}[Safety]
\label{safe}
The trajectories of a given system are safe with respect to a set $\mathcal{C}$ if $\mathcal{C}$ is forward invariant, meaning that for every $x(0) \in \mathcal{C}$, $x(t) \in \mathcal{C}$ for all $t>0$.
\end{defn} 
\vspace{-2mm}
Consider $N$ subsets $\mathcal{C}_1, \dots, \mathcal{C}_N \subset \mathbb{R}^n$  defined by the superlevel set of a continuously differentiable function $h_i : \mathbb{R}^n \rightarrow \mathbb{R}$:
\begin{align}
\mathcal{C}_i &= \{x \in \mathbb{R}^n: h_i(x) \ge 0\},\quad i =1,2,\ldots, N \label{eq:boundary}
\end{align}
\vspace{-8mm}
\begin{defn}[CBFs]
Let $\mathcal{C}_i \in \mathbb{R}^n$ be defined by one of the functions \eqref{eq:boundary}.
Then $h_i(x)$ is a (zeroing) {\it Control Barrier Function} (CBF) for \eqref{eq:affine_nonlinear} if there exists a locally Lipschitz extended class 
$\mathcal{K}_{\infty}$ function\footnote{{ A class $\mathcal{K}_{\infty}$ function $\alpha: \mathbb{R} \rightarrow \mathbb{R}$ is strictly increasing with 
$\alpha(0) \!=\! 0$ and $\lim_{t \rightarrow \infty} \alpha(t) = \infty$.}} $\alpha$ such that
\begin{align}
\sup_{u\in \mathbb{R}^m} \left[L_f h_i(x) + \inner{L_g h_i(x)}{u}{} \right] \ge -\alpha(h_i(x)) \quad \forall x \in \mathbb{R}^n \,. \nonumber
\end{align}
\end{defn}
\vspace{-4mm}
This definition simply means that the CBFs $h_i(x)$ are only allowed to decrease in the interior of their respective safe sets $\text{int}(\mathcal{C}_i)$, but not on their boundaries $\partial \mathcal{C}_i$.

Consider the closed-loop system for \eqref{eq:affine_nonlinear} 
\begin{align}
\dot{x} = f_{cl}(x) := f(x) + g(x) u^\star(x)
\label{eq:feedback_system}
\end{align}
with control law $u^\star(x)$ given by the minimum-norm feedback controller based on \cite{Ames2014}
\begin{align}
u^\star(x) &= \argmin_{(u,\delta)\in\mathbb{R}^{m+1}} \frac{1}{2} \norm{u}{}^2 + \frac{1}{2} p \delta^2 \qquad 
\label{eq:QP_control} \\
s.t. \,
L_f V &+ \inner{L_g V}{u}{} + \gamma(V) \le \delta \nonumber 
\\
\quad L_f h_i &+ \inner{L_g h_i}{u}{} \ge - \alpha(h_i) \,, \quad i \in \{ 1, \cdots, N \} \nonumber
\end{align}
with $p > 0$, $\gamma$ and $\alpha$ being class $\mathcal{K}$ and class $\mathcal{K}_\infty$ functions, respectively.
%
%
If feasible, the feedback controller \eqref{eq:QP_control} guarantees local stability of $x_0$ and safety of the closed-loop system trajectories with respect to the safe set 
\begin{align}
\mathcal{C} = \cap^N_{i=1} \mathcal{C}_i
\label{eq:joint_safe_set}
\end{align}
However, \eqref{eq:QP_control} does not guarantee global stabilization, meaning that the trajectories could converge towards equilibrium points other than the CLF minimum \cite{Reis_LCSS}.
%
%
\begin{assum}
\label{assumption:initial_state}
The initial state $x(0) \in \mathbb{R}^n$ is contained in the safe set \eqref{eq:joint_safe_set} and
the CLF minimum $x_0$ is contained in $\mathcal{C}$, that is, $h_i(x_0) \ge 0$ for all $i \in \{ 1, \cdots, N \}$.
\end{assum}
Assumption \ref{assumption:initial_state} comes from the fact that it is natural to assume that a system starts in a safe configuration; as an example, it is only natural to assume that a vehicle starts its navigation task in a safe state of non-collision against obstacles. 
%
%
Furthermore, the CLF minimum $x_0$ must be reachable by the controller \eqref{eq:QP_control}.
\begin{thm}
\label{theorem0}
Under \aref{assumption:initial_state}, the QP \eqref{eq:QP_control} is feasible for all $x \in \mathcal{C}$, if at least one the two conditions are met:
\begin{itemize}
    \item[(i)] There is only one CBF ($N=1$).
    \item[(ii)] The affine non-linear system \eqref{eq:affine_nonlinear} is driftless, that is, $f(x) = 0 \,\,\, \forall x \in \mathbb{R}^n$.
\end{itemize}
\end{thm}

\vspace{-4mm}
\begin{pf}
The proof for {\bf(i)} was first introduced in \cite{Ames2014}.
The proof of {\bf(ii)} is as follows: under Assumption \ref{assumption:initial_state}, the initial state $x(0) \in \mathcal{C}$, that is, $h_i(x(0)) \ge 0$ for all $i=1,\dots, N$.
Then, for driftless affine nonlinear systems, the decision space associated to the $i$-th CBF constraint is given by the half-plane $\mathbb{K}_{cbf_i}(x) = \{ (u,\delta) \in \mathbb{R}^{m+1} : 
\inner{L_g h_i(x)}{u}{} + \alpha h_i(x) \ge 0 \}$.
The intersection of these half-planes configures a convex polytope $\mathbb{K}_{cbf}(x) = \cap^N_{i=1} \mathbb{K}_{cbf_i}(x)$, which is the decision space associated to the CBF constraints. Due to the independence of the CBF constraints on the slack variable $\delta$
and due to the fact that $h_i(x(0)) \ge 0$, $\mathbb{K}_{cbf}(x(0))$ contains the the entire $\delta$-axis, that is, the line $(u, \delta) = (0, \delta) \,\, \forall \delta \in \mathbb{R}$. Therefore, $\mathbb{K}_{cbf}(x(0))$ is unbounded and non-empty.
Defining the decision space associated to the CLF constraint as the half-plane $\mathbb{K}_{clf}(x) = \{ (u,\delta) \in \mathbb{R}^{m+1}: \inner{L_g V(x)}{u}{} + \gamma(V(x)) \le \delta \}$, the feasible set associated to the QP \eqref{eq:QP_control} is the intersection 
$\mathbb{K}_{clf}(x) \cap \mathbb{K}_{cbf}(x) \subset \mathbb{R}^{m+1}$.
Notice that $\mathbb{K}_{clf}(x(0)) \cap \mathbb{K}_{cbf}(x(0)) \neq \emptyset$, meaning that the QP is initially feasible under Assumption \ref{assumption:initial_state}.
Then, since the CBF constraints guarantee the invariance of the trajectories $x(t)$ with respect to the safe set $\mathcal{C}$, $h_i(x(t)) \ge 0 \forall t \ge 0$, $i=1,\cdots,N$. Therefore, the convex polytope $\mathbb{K}_{cbf}(x(t))$ remains unbounded and non-empty $\forall t \ge 0$ (it must always contain the $\delta$-axis), and therefore the feasible set for the QP \eqref{eq:QP_control} is 
$\mathbb{K}_{clf}(x(t)) \cap \mathbb{K}_{cbf}(x(t)) \neq \emptyset$ for all $t \ge 0$. \myqed
\end{pf}
%
%
%
%
\vspace{-4mm}
\begin{assum}[Disjoint Unsafe Sets]
\label{assumption:disjoint_barriers}
The unsafe sets of the $N$ barriers are disjoint, that is:
\begin{align}
    \overline{\mathcal{C}}_i \cap \overline{\mathcal{C}}_j = \emptyset \qquad \forall i \ne j 
    \label{eq:disjoint_barriers}
\end{align}
\end{assum}
\vspace{-5mm}
\begin{rem}
Assumption \ref{assumption:disjoint_barriers} is not restrictive for the following reason: assume there exist barriers $h_i$, $h_j$ with non-empty unsafe 
set intersections, that is, $\overline{\mathcal{C}}_i \cap \overline{\mathcal{C}}_j \neq \emptyset$. Then, it is possible to construct a new composite barrier 
$h_{k}$ with unsafe set $\overline{\mathcal{C}}_k \supset \overline{\mathcal{C}}_i \cup \overline{\mathcal{C}}_j$ (under mild assumptions on the regularities 
of $h_{k}$), thus representing (almost) the same safe region as $\mathcal{C}_i \cap \mathcal{C}_j$. \cite{Molnar2023} proposes a composition method for 
combining multiple CBFs into a single one.
\end{rem}
%
%
\begin{defn}
\label{eq:transCLF}
Given a CLF $V: \mathbb{R} \rightarrow \mathbb{R}_{\ge 0}$, define the transformed CLF $\overline{V}: \mathbb{R} \rightarrow \mathbb{R}_{\ge 0}$ as
\begin{align}
\overline{V}(x) &= \int^{V(x)}_0 \gamma( \tau ) d \tau \label{eq:barV} \\
\nabla \overline{V} &= \gamma(V) \nabla V \label{eq:nabla_barV} \\
H_{\overline{V}} &= \gamma(V) H_V + \gamma^{\prime}(V) \nabla V \nabla V^\mathsf{T} \label{eq:H_barV}
\end{align}
\end{defn}
\begin{prop}
\label{prop:properties_transformed_CLF}
The transformed CLF \eqref{eq:barV} has the following properties:
\begin{itemize}
\item[(i)] $\overline{V}(x) > 0$ $\forall x \ne x_0$. Additionally, $\overline{V}(x_0) = 0$.
\item[(ii)] The integral transformation of \eqref{eq:barV} is invertible.
\item[(iii)] $V(x)$ and $\overline{V}(x)$ have the same level sets.
\end{itemize}
\end{prop}
\vspace{-5mm}
\begin{pf}
Property (i) can be seen from the fact that $\gamma: \mathbb{R}_{\ge 0} \rightarrow \mathbb{R}_{\ge 0}$ is a class 
$\mathcal{K}$ function, and therefore its integral is positive and strictly increasing. 
Furthermore, at $x = x_0$, $V(x_0) = 0$ and the limits of integration on \eqref{eq:barV} are both zero, showing that $\overline{V}(x_0) = 0$.
Property (ii) can also be inferred from the fact that $\gamma$ is of class $\mathcal{K}$: since its integral is a positive and strictly increasing function, its inverse always exists.
That means that the original $V$ can always be computed from the transformed CLF $\overline{V}$ by inverting the integral transformation \eqref{eq:barV}.
Property (iii) holds because by \eqref{eq:nabla_barV}, the gradients of $V$ and $\overline{V}$ are co-directed and $V$, $\overline{V}$ are continuous functions.
Therefore, they must share the same level sets. \myqed 
\end{pf}
\vspace{-6mm}
%
%
As will be shown in the next sections, the CLF transformation in \dref{eq:transCLF} will be useful not only for expressing the existence and stability 
conditions for equilibrium points in a simpler way, but also for developing the method for CLF-compatibility that is presented in \sref{sec:clf_compatibility}. 
\section{Existence of Equilibrium Points in the CLF-CBF Framework}
\label{sec:existence}
In this section, we extend a result from \cite{Reis_LCSS}, regarding the existence of equilibrium points when multiple CBF constraints are present.
%

\begin{defn}[Equilibrium Manifold]
\label{def:equilibrium_manifold}
Define the vector-valued transformation 
$f_i: \mathbb{R}^n \times \mathbb{R}_{\ge0} \rightarrow \mathbb{R}^{n \times n}$ associated to the $i$-th CBF as
\begin{align}
f_{i}(x, \lambda) = f + \lambda G \nabla h_{i} - p G \nabla \overline{V}
\label{eq:fi}
\end{align}
where $G(x) = g(x) g(x)^\mathsf{T} \in \mathbb{R}^{n \times n}$.
The Jacobian matrix of $f_i$ with respect to $x \in \mathbb{R}^n$ is
\begin{align}
J_{f_i}(x, \lambda) 
&= \frac{\partial f}{\partial x}
+ \lambda \frac{\partial G \nabla h_i}{\partial x}
- p \frac{\partial G \nabla \overline{V}}{\partial x}
\label{eq:Jfi}
\end{align}
\end{defn}
\vspace{-4mm}
As will be demonstrated in the next sections, \eqref{eq:fi}-\eqref{eq:Jfi} will be of central importance to characterize the existence and stability conditions for the equilibrium points of the closed-loop system.
\begin{thm}[Existence of Equilibrium Points]
\label{theorem:existence_equilibria}
Let \eqref{eq:feedback_system} be the closed-loop system formed by combining the nonlinear system \eqref{eq:affine_nonlinear}, under Assumption \ref{assumption:disjoint_barriers}, with the controller \eqref{eq:QP_control}.
{
The set of equilibrium points of \eqref{eq:feedback_system} is given by 
$\mathcal{E} = \left( \bigcup^N_{i=1} \mathcal{E}_{\partial \mathcal{C}_i} \right) \cup \mathcal{E}_{int(\mathcal{C})}$,
%
%
with
\begin{align}
    \mathcal{E}_{\partial \mathcal{C}_i} &= \{ x \in \Omega^{clf}_{i} \cap \partial \mathcal{C}_i \,|\, f_i(x, \lambda_i) = 0 \} 
    \label{eq:boundary_equilibrium} \\
    \mathcal{E}_{int(\mathcal{C})} &= \{ x \in \Omega^{clf}_{\overline{cbf}} \cap int(\mathcal{C}) \,|\, f_i(x, 0) = 0 \} 
    \label{eq:interior_equilibrium}
\end{align}
where $\Omega^{clf}_{i} \subset \mathbb{R}^n$ is the set of states where the CLF and only the $i$-th CBF constraint are active, and
$\Omega^{clf}_{\overline{cbf}} \subset \mathbb{R}^n$ is the set of states where the CLF constraint is active and all CBF constraints are inactive. Moreover, $\lambda_i \ge 0$ in \eqref{eq:boundary_equilibrium} is the KKT multiplier associated to the $i$-th active CBF constraint. The set $\mathcal{E}_{\partial \mathcal{C}_i}$ in \eqref{eq:boundary_equilibrium} is the set of {\it boundary} equilibrium points, while $\mathcal{E}_{int(\mathcal{C})}$ in \eqref{eq:interior_equilibrium} is the set of {\it interior} equilibrium points.
}
%
%
%
\end{thm}
\vspace{-4mm}
\begin{pf}
The Lagrangian associated to QP \eqref{eq:QP_control} is
\begin{align}
\mathcal{L} = \frac{1}{2} \left( \norm{u}{}^2 + p \delta^2 \right)
+ \lambda_0&( F_V + L_g V u - \delta ) \nonumber \\
- \sum^N_{i=1}&\lambda_i ( F_{h_i} + L_g h_i \, u )
\label{eq:general_lagrangian}
\end{align}
where $F_V(x) = L_f V + \gamma(V)$ and $F_{h_i}(x) = L_f h_i + \alpha(h_i)$, and 
$\lambda_i \ge 0 \in \{0, 1, \cdots, N\}$ are the KKT multipliers associated to the optimization problem.
Then, the KKT conditions are:
\begin{align}
\frac{\partial \mathcal{L}}{\partial u} 
= u + \lambda_0 g^\mathsf{T} \nabla V - \sum^{N}_{i=1} \lambda_i g^\mathsf{T} \nabla h_i &= 0 \label{eq:KKT1} \\
\frac{\partial \mathcal{L}}{\partial \delta} = p \delta - \lambda_0 &= 0 \label{eq:KKT2} \\
\lambda_0 ( F_{V} + L_g V \,u - \delta ) &= 0 \label{eq:KKT3} \\
\lambda_i ( F_{h_i} + L_g h_i \,u ) &= 0 \label{eq:KKT4}
\end{align}
with $i \in \{ 1, \cdots, N \}$. Using \eqref{eq:KKT1}-\eqref{eq:KKT2}, the QP solutions are given by:
\begin{align}
u^\star(x) &= g^\mathsf{T} \left( -\lambda_0 \nabla V + \sum^N_{i=1} \lambda_i \nabla h_i \right)
\label{eq:control_solution} \\
\delta^\star(x) &= p^{-1} \lambda_0 \,,
\label{eq:delta_solution}
\end{align}
with $\lambda_i \ge 0 \in \{ 0, 1, \cdots, N \}$. Substituting \eqref{eq:control_solution} on \eqref{eq:feedback_system} yields the following expression for the closed-loop system:
\begin{align}
f_{cl}(x) \!=\! f \!+\! G \left(- \lambda_0 \nabla V \!+\! \sum^N_{i=1} \lambda_i \nabla h_i \right) \,, 
\label{eq:closed_loop}
\end{align}
where $G(x) = g(x) g(x)^\mathsf{T}$ is a positive semi-definite matrix.
At an equilibrium point $x_e \in \mathcal{E}$, $f_{cl}(x_e) = 0$. Applying this condition to \eqref{eq:closed_loop} yields
\begin{align}
f(x_e) = G(x_e) \left( \lambda_0 \nabla V(x_e) - \sum^N_{i=1} \lambda_i \nabla h_i(x_e) \right)
\label{eq:equilibrium_condition}
\end{align}
%
%
Consider the region of the state space where the CLF constraint is {\it inactive}: 
$L_{f_{cl}(x)} V(x) + \gamma( V(x) ) - \delta^\star(x) < 0$. 
From \eqref{eq:KKT3}, $\lambda_0 = 0$. Then, using \eqref{eq:delta_solution}, notice that $\delta^\star(x) = 0$. 
At an equilibrium point $x_e \in \mathcal{E}$, $L_{f_{cl}(x_e)} V(x_e) = 0$, and therefore we obtain 
$\gamma( V(x_e)) < \delta^\star(x_e) = 0$, implying that 
$V(x_e) < 0$, which is a contradiction since $V$ is a nonnegative function. Therefore, all equilibrium points must lie on regions where the CLF constraint is {\it active}.\\
Consider the region where CLF constraint is {\it active}: 
$L_{f_{cl}(x)} V(x) + \gamma(V(x)) = \delta^\star(x)$. At an equilibrium point $x_e \in \mathcal{E}$, $L_{f_{cl}(x_e)} V(x_e) = 0$. Therefore, using \eqref{eq:delta_solution}, $\gamma(V(x_e)) = \delta^\star(x_e) = p^{-1} \lambda_0$.
Then, at any equilibrium point $x_e \in \mathcal{E}$, the KKT multiplier associated to the CLF constraint is 
$\lambda_0(x_e) = p \gamma(V(x_e)) \ge 0$.
Therefore, equation \eqref{eq:equilibrium_condition} yields:
\begin{align}
f(x_e) \!=\! G(x_e) \!\left(\! p \nabla \overline{V}(x_e) \!-\! \sum^N_{i=1} \lambda_i \nabla h_i(x_e) \!\right)
\label{eq:explicit_equilibrium_condition}
\end{align}
where $\nabla \overline{V}(x_e) = \gamma(V(x_e)) \nabla V(x_e)$.
For the next two cases, the CLF constraint is assumed to be active.\\
{\bf Case 1.} Consider the region $\Omega^{clf}_{i} \subset \mathbb{R}^n$, where the CLF constraint is active and only the $i$-th CBF constraint is {\it active}: $L_{f_{cl}(x)} h_i(x) + \alpha(h_i(x)) = 0$. At $x_e \in \mathcal{E}$, $L_{f_{cl}(x_e)} h_i(x_e) = 0$, implying that $h_i(x_e) = 0$. Therefore, equilibrium points occurring in this region must lie on the boundary of the $i$-th safe set, that is, $x_e \in \partial \mathcal{C}_i$.
%
Next, we show that, under Assumption \ref{assumption:disjoint_barriers}, these equilibrium points can only occur when {\it only} the $i$-th CBF constraint is active.
Assume that $x_e \in \mathcal{E}$ occurs when two CBF constraints are active: that is, we have 
$L_{f_{cl}(x_e)} h_i(x_e) + \alpha(h_i(x_e)) = 0$ and 
$L_{f_{cl}(x_e)} h_j(x_e) + \alpha(h_j(x_e)) = 0$, for $i \ne j$, 
and therefore, since $L_{f_{cl}(x_e)} h_i(x_e) = 0$ and $L_{f_{cl}(x_e)} h_j(x_e) = 0$, we have
$h_i(x_e) = h_j(x_e) = 0$.
However, by Assumption \ref{assumption:disjoint_barriers}, this is a contradiction since $\partial \mathcal{C}_i \cap \partial \mathcal{C}_j = \emptyset$. 
The conclusion is that boundary equilibrium points on the $i$-th boundary are located in the set where {\it only} the $i$-th CBF constraint is active, denoted by $\Omega^{clf}_{i}$. That means that at $x_e \in \partial \mathcal{C}_i$, $\lambda_j = 0 \,, \,\, \forall j \ne i$.
Therefore, \eqref{eq:explicit_equilibrium_condition} reduces to 
\begin{align}
f(x_e) = G(x_e) \left( p \nabla \overline{V}(x_e) - \lambda_i \nabla h_i(x_e) \right)
\label{eq:practical_equilibrium_condition}
\end{align}
where $\lambda_i \ge 0$ is the KKT multiplier associated to the active CBF. Notice that \eqref{eq:practical_equilibrium_condition} is equivalent to $f_i(x,\lambda_i) = 0$, with $f_i$ defined by \eqref{eq:fi}.
Thus, in this case, the equilibrium point is on the boundary of the safe set and satisfies $f_i(x_e, \lambda_i) = 0$ 
for some $\lambda_i \ge 0$, demonstrating \eqref{eq:boundary_equilibrium}.\\
%
{\bf Case 2.} Consider the region $\Omega^{clf}_{\overline{cbf}} \subset \mathbb{R}^n$, where the CLF constraint is active and all CBF constraints are {\it inactive}: $L_{f_{cl}(x)} h_i(x) + \alpha( h_i(x) ) > 0$. From \eqref{eq:KKT4}, $\lambda_1 = \cdots = \lambda_N = 0$.
At an equilibrium point $x_e \in \mathcal{E}$, $L_{f_{cl}(x_e)} h_i(x_e) = 0$, implying that $h_i(x_e) > 0$.
Therefore, equilibrium points occurring in this region must lie in the interior of the safe set, that is, $x_e \in int(\mathcal{C})$.
Additionally, \eqref{eq:explicit_equilibrium_condition} must be satisfied with $\lambda_1 = \cdots = \lambda_N = 0$, which means that $f(x_e) = p G(x_e) \nabla \overline{V}(x_e)$, which is equivalent to $f_i(x_e, 0) = 0$. This demonstrates \eqref{eq:interior_equilibrium}. \myqed 
%
%
\end{pf}
\vspace{-5mm}
A similar version of \tref{theorem:existence_equilibria} was demonstrated in \cite{Reis_LCSS}, considering only one CBF. 
Therefore, combining stabilization and safety objectives with the CLF-CBF framework can introduce equilibrium points in 
the closed-loop system other than the CLF global minimum $x_0 \in \mathbb{R}^n$, some of them could even possibly be 
asymptotically stable \cite{Reis_LCSS}.
This is a known problem in CLF-CBF literature and was considered in other works as well, such as in \cite{tan2024undesired}, 
which has presented a similar characterization of the equilibrium points and has proposed a modified QP-based controller for \eqref{eq:affine_nonlinear}: 
$u(x) = \overline{u}(x) + u_{\text{nom}}(x)$, where $\overline{u}(x), u_{\text{nom}}(x) \in \mathbb{R}^m$ are feedback controllers to be designed as follows. 
Substituting $u(x)$ into \eqref{eq:affine_nonlinear} transforms the system dynamics into:
\begin{align}
\dot{x} &= \overline{f}(x) + g(x) \overline{u}(x)
\label{eq:transformed_affine}
\end{align}
where $\overline{f}(x) = f(x) + g(x) u_{\text{nom}}(x)$.
\vspace{-2mm}
\begin{assum}[CLF Condition]
\label{assumption:clf_condition}
Given a nonlinear control affine dynamical system \eqref{eq:affine_nonlinear}, the CLF $V(x)$ satisfies 
$L_{f} V(x) < 0 \,, \forall x \neq x_0 \in \mathbb{R}^n$.
\end{assum}
Let $u_{\text{nom}}(x)$ be a feedback controller chosen in such a way that \aref{assumption:clf_condition} holds for the CLF $V$ and system \eqref{eq:transformed_affine}.
This can always be done if the system \eqref{eq:affine_nonlinear}, { provided that it admits a valid CLF}:\\
{\bf (i)} In case $V$ satisfies \aref{assumption:clf_condition} for the original system \eqref{eq:affine_nonlinear}, $u_{\text{nom}}(x) = 0$ and $\overline{f}(x) = f(x)$.\\
{\bf (ii)} In case $V$ does not satisfy \aref{assumption:clf_condition} for the original system \eqref{eq:affine_nonlinear}, one can find $u_{\text{nom}}(x)$ such that $V$ satisfies \aref{assumption:clf_condition} for $\overline{f}(x) = f(x) + g(x) u_{\text{nom}}(x)$ of the transformed system \eqref{eq:transformed_affine}, for example, by using Sontag's formula \cite{Sontag1989}.\\
In \cite{tan2024undesired}, it was shown that with $u_{\text{nom}}(x)$ chosen in this way and with $\overline{u}(x)$ obtained from solving the QP \eqref{eq:QP_control} 
with the system model given by the transformed dynamics \eqref{eq:transformed_affine}, then the closed-loop system obtained from applying controller 
$u(x) = \overline{u}(x) + u_{\text{nom}}(x)$ into \eqref{eq:affine_nonlinear} has the following set of equilibrium points:
\begin{align}
\overline{\mathcal{E}} &= \left( \bigcup^N_{i=1} \overline{\mathcal{E}}_{\partial \mathcal{C}_i} \right) \cup \{x_0\} 
\label{eq:transf_equilibrium_set} \\ 
\overline{\mathcal{E}}_{\partial \mathcal{C}_i} &= \mathcal{E}_{\partial \mathcal{C}_i} \cap \{ x \in \mathbb{R}^n \,|\, L_g h_i(x) \ne 0 \} \nonumber
\end{align}
That is, interior equilibrium points other than the CLF minimum and boundary equilibrium points satisfying $L_g V(x) = 0$ do not exist.
However, the existence of boundary equilibrium points with $L_g V(x) \ne 0$ is not excluded.
Since $\overline{u}(x)$ is obtained through solving the QP \eqref{eq:QP_control} for the a new nonlinear control affine system 
\eqref{eq:transformed_affine}, the theory developed so far for the remaining closed-loop equilibrium points remains valid.
%
\begin{prop}
\label{prop:trans_clf_condition}
If \aref{assumption:clf_condition} holds for a CLF $V(x)$, it also holds for the transformed CLF $\overline{V}$ from \dref{eq:transCLF}.
\end{prop}
\vspace{-4mm}
\begin{pf}
If \aref{assumption:clf_condition} holds for $V(x)$, then $L_f V = \inner{\nabla V}{f}{} < 0 \, \forall x \neq x_0$. 
Using \eqref{eq:nabla_barV}, for all $x \neq x_0$ we have $\nabla V = \gamma(V)^{-1} \nabla \overline{V}$ with $\gamma(V)^{-1} > 0$.
Then, $L_f V = \gamma(V)^{-1} \inner{\nabla \overline{V}}{f}{} < 0$ implies $L_f \overline{V} < 0$ for all $x \neq x_0$. \myqed 
\end{pf}
%
%

\section{Stability of Equilibrium Points in the CLF-CBF Framework}
\label{sec:stability}

{ Assuming the system admits a valid CLF}, one can always apply the technique proposed in \cite{tan2024undesired} and work with the transformed system \eqref{eq:transformed_affine}. Therefore, from this point on, we assume to be working with a nonlinear control affine system such that \aref{assumption:clf_condition} is satisfied.
Consequently, only the remaining boundary equilibrium points with $L_g h(x) \neq 0$ need to be addressed.
%
%
Our objective in this section is to study the stability properties of these points.
Particularly, we generalize \cite[Theorem 2]{Reis_LCSS} for nonlinear control affine systems, deriving a necessary and sufficient condition for 
the instability of boundary equilibrium points satisfying $L_g h(x) \neq 0$, when multiple CBF constraints are present.

\begin{defn}
\label{def:z_vectors}
Define the two vector fields $z_1, z_2: \mathbb{R}^n \rightarrow \mathbb{R}^n$ associated to the $i$-th barrier as
\begin{align}
z_1(x) &= \frac{\nabla h_i}{\norm{\nabla h_i}{G}}\,, 
\label{eq:z1} \\
z_2(x) &= \nabla V - \inner{\nabla V}{z_1}{G} \, z_1 \label{eq:z2} \,,
\end{align}
where $G(x) = g(x) g(x)^\mathsf{T}$. 
\end{defn}
One can verify that $\{z_1, z_2\}$ is an orthogonal set of vectors with respect to the inner product $\inner{\cdot}{\cdot}{G}$, that is, 
$\inner{z_i}{z_j}{G} = 0$, $i,j \in \{1,2\}$, $i \ne j$. In particular, $\inner{z_1}{z_1}{G} = 1$.
Furthermore, define the scalar function $\eta: \mathbb{R}^n \rightarrow \mathbb{R}_+$ as
\begin{align}
\eta(x) &= (1 + p \inner{z_2}{z_2}{G})^{-1}
\label{eq:eta}
\end{align}
with the following properties:
\begin{enumerate}[label=(\roman*)]
\item $0 < \eta(x) \le 1 \,\, \forall x \in \mathbb{R}^n$, since $\inner{z_2}{z_2}{G} \ge 0$.
\item $\eta = 1$ if and only if (i) $z_2 = 0$ or (ii) $G z_2 = 0$.
\item From \eqref{eq:eta}, the inner product $\inner{z_2}{z_2}{G}$ can be expressed as $\inner{z_2}{z_2}{G} = p^{-1} (\eta^{-1} - 1) \ge 0$,
\item Combining \eqref{eq:z2} and (iii), it is possible to demonstrate
$p^{-1} + \norm{\nabla V}{G}^2 = \inner{\nabla V}{\hat{z}_1}{G}^2 + p^{-1} \eta^{-1}$.
\end{enumerate}
\begin{lem}[Boundary Jacobian]
\label{lemma:jacobian_boundary}
Under Assumption \ref{assumption:disjoint_barriers}, the Jacobian matrix $J_{cl}(x_e) \in \mathbb{R}^{n \times n}$ of the closed-loop system \eqref{eq:feedback_system} 
computed at a boundary equilibrium point $x_e \in \mathcal{E}_{\partial \mathcal{C}_i}$ with $L_g h_i(x_e) \ne 0$ and corresponding KKT multiplier $\lambda_i \ge 0$ is given by
\begin{align}
J_{cl}(x_e) &\!=\! \left( I_n \!-\! G Z N_1 Z^\mathsf{T} \right) J_i(x_e, \lambda_i)
\!-\! G Z N_1 \Psi Z^\mathsf{T}
\label{eq:jacobian_boundary}
\end{align}
where $Z(x) = [\, z_1 \,\,\, z_2 \,] \in \mathbb{R}^{n \times 2}$ with $z_i(x)$ as defined in \eqref{eq:z1}-\eqref{eq:z2}, 
$N_1(x) = \Diag{1, p \eta(x)} > 0$ with $\eta(x)$ as defined in \eqref{eq:eta}. 
%
Matrices $J_{i}$ and $\Psi$ are given by
\begin{align}
J_{i}(x, \lambda) &= \frac{\partial f(x)}{\partial x} + \lambda_i \frac{\partial G \nabla h_i}{\partial x} - p \gamma(V) \frac{\partial G \nabla V}{\partial x}
\label{eq:Ji} \\
\Psi &= \begin{bmatrix}
\alpha^{\prime}(h_i) & 0 \\
\inner{\nabla V}{\hat{z}_1}{G} (\gamma^{\prime}(V) - \alpha^{\prime}(h_i)) & \gamma^{\prime}(V)
\end{bmatrix} \label{eq:Psi} 
\end{align}
\end{lem}
\vspace{-7mm}
\begin{pf}
\label{sec:proof_jacobian_boundary}
This demonstration is a direct continuation of the proof of Theorem \ref{theorem:existence_equilibria}.
In {\bf Case 3}, one can substitute \eqref{eq:control_solution} and \eqref{eq:delta_solution} into the complementary 
slackness conditions \eqref{eq:KKT3}-\eqref{eq:KKT4} and use the fact that $\lambda_j = 0$ for all $j \ne i$ to get 
the following system:
\begin{align}
\underbrace{\begin{bmatrix}
p^{-1} + \norm{\nabla V}{G}^2 \!&\! - \inner{\nabla V}{\nabla h_i}{G} \\
-\inner{\nabla V}{\nabla h_i}{G} \!&\! \norm{\nabla h_i}{G}^2
\end{bmatrix}}_{C(x)}
\underbrace{\begin{bmatrix}
\lambda_0 \\
\lambda_i
\end{bmatrix}}_{\bar{\lambda}_i}
\!=\!
\underbrace{\begin{bmatrix}
F_V \\
- F_{h_i}
\end{bmatrix}}_{b(x)}
\label{eq:lambda_solutions}
\end{align}
where $F_V = L_f V + \gamma(V)$ and $F_{h_i} = L_f h_i + \alpha(h_i)$. 
In particular, the set $\Omega^{clf}_{i}$ where the CLF and only the $i$-th CBF constraint is active is given by
\begin{align}
\Omega^{clf}_{i} 
= \{ x \in \mathbb{R}^n \,|\, & f_{cl}(x) \!=\! f \!+\! G \left( \lambda_i \nabla h_i - \lambda_0 \nabla V \right) \}
\label{eq:Si_set}
\end{align}
where $\lambda_0, \lambda_i$ are the solutions of \eqref{eq:lambda_solutions}.
%
%
The determinant of $C(x)$ is given by 
$|C(x)| = ( p^{-1} + \norm{\nabla V}{G}^2 )\norm{\nabla h_i}{G}^2 - \inner{\nabla V}{\nabla h_i}{G}^2$, 
which can be simplified by using the  definition of $z_2$ in \eqref{eq:z2}, yielding 
$|C(x)| = (p \eta)^{-1} \norm{\nabla h_i}{G}^2$. Notice that $|C(x)| \ge 0$, being zero if and only if $L_g h_i(x) = 0$.
Since we are considering the case $L_g h_i(x) \neq 0$, an expression for the inverse of $C(x)$ is then given by
\begin{align}
C(x)^{-1} &= \frac{p \eta}{\norm{\nabla h_i}{G}^2}
\begin{bmatrix}
\norm{\nabla h_i}{G}^2 & \inner{\nabla V}{\nabla h_i}{G} \\
\inner{\nabla V}{\nabla h_i}{G} & p^{-1} + \norm{\nabla V}{G}^2
\end{bmatrix}
\label{eq:inverseA}
\end{align}
The derivative of \eqref{eq:lambda_solutions} with respect to the $k$-th state component $x_k$ yields
\begin{align}
\partial_k C(x) \bar{\lambda}_i(x) \!+\! C(x) \partial_k \bar{\lambda}_i(x) = \partial_k b(x)
\label{eq:linear_system_derivative}
\end{align}
where the operator $\partial_k$ denotes the partial derivative with respect to $x_k$. In the case where 
$L_g h(x) \ne 0$, $|C(x)| \neq 0$ and therefore the inverse \eqref{eq:inverseA} can be used to directly 
solve \eqref{eq:linear_system_derivative} for the partial derivatives of $\partial_k \bar{\lambda}_i(x)$:
\begin{align}
\partial_k \bar{\lambda}_i(x) &= C(x)^{-1} \underbrace{ \left( \partial_k b(x) - \partial_k C(x) \bar{\lambda}_i(x) \right) }_{c(x)}
\label{eq:partial_lambda} 
\end{align}
To find expressions for $\partial_k \lambda_0(x)$, $\partial_k \lambda_i(x)$ using \eqref{eq:partial_lambda}, expressions for 
$\partial_k C(x)$ and $\partial_k b(x)$ must be derived. 
Matrix $\partial_k C(x)$ is dependent on $\partial_k\norm{\nabla V}{G}$, 
$\partial_k\norm{\nabla h_i}{G}$ and $\partial_k\inner{\nabla V}{\nabla h_i}{G}$.
Vector $\partial_k b(x)$ is dependent on $\partial_k F_V$ and $\partial_k F_{h_i}$.
These can be computed using the derivatives
\begin{align}
\partial_k\inner{u}{v}{G} &\!=\! \inner{\partial_k u}{v}{G} \!+\! \inner{\partial_k v}{u}{G} \!+\! u^\mathsf{T}(\partial_k G)v
\label{eq:partial_inner} \\
\partial_k F_w &\!=\! \inner{\partial_k \nabla w}{f}{} \!+\! \inner{\nabla w}{\partial_k f}{} \!+\! \beta^{\prime}(w) [\nabla w]_k
\label{eq:partial_F}
\end{align}
with $u(x),v(x) \in \mathbb{R}^n$ replaced by $\nabla V(x)$ or $\nabla h_i(x)$ in \eqref{eq:partial_inner}, 
and with $w(x) \in \mathbb{R}$ replaced by $V(x)$ or $h_i(x)$ 
and $\beta$ replaced by $\gamma$ of class $\mathcal{K}$ or $\alpha$ of class $\mathcal{K}_\infty$ in \eqref{eq:partial_F}.
Using \eqref{eq:partial_inner}-\eqref{eq:partial_F}, an expression for $c(x)$ can be found:
\begin{align}
c(x) &=
\begin{bmatrix}
\inner{\partial_k \nabla V}{f_{cl}}{} \!+\! \inner{\nabla V}{j_{k}}{} \!+\! \gamma^{\prime}\!(V) \partial_k V \\
- \inner{\partial_k \nabla h_i}{f_{cl}}{} \!-\! \inner{\nabla h_i}{j_{k}}{} \!-\! \alpha^{\prime}\!(h_i) \partial_k h_i \\
\end{bmatrix} 
\label{eq:c_vector} \\
j_{k} &= \partial_k f + \lambda_i \partial_k(G \nabla h_i) - \lambda_0 \partial_k (G \nabla V) \nonumber
\end{align}
At a boundary equilibrium point $x_e \in \mathcal{E}_{\partial \mathcal{C}_i}$, $f_{cl}(x_e) = 0$ and $\lambda_0 = p \gamma(V(x_e))$, 
simplifying \eqref{eq:c_vector} and allowing \eqref{eq:partial_lambda} to be written as
\begin{align}
\partial_k \bar{\lambda}_i(x_e) &\!=\! 
C(x)^{-1} \!\!
\begin{bmatrix}
\nabla V^\mathsf{T} [J_i(x_e, \lambda_i)]_k
\!+\! \gamma^{\prime}\!(V) \partial_k V 
\\
- \nabla h_i^\mathsf{T} [J_i(x_e, \lambda_i)]_k
\!-\! \alpha^{\prime}\!(h_i) \partial_k h_i
\end{bmatrix}
\label{eq:partial_lambda2}
\end{align}
where $[J_i(x_e, \lambda_i)]_k$ denotes the $k$-th column of matrix $J_i(x_e, \lambda_i)$ defined at \eqref{eq:Ji}.
Using $C(x_e)^{-1}$ from \eqref{eq:inverseA}, expressions for $\partial_k \lambda_0(x_e)$ and $\partial_k \lambda_i(x_e)$ follow from \eqref{eq:partial_lambda2}. 

Equation \eqref{eq:closed_loop} with $\lambda_j = 0\,, \forall j \ne i$ gives the closed-loop system expression for {\bf Case 3}, 
which is $f_{cl}(x) = f(x) - \lambda_0 G(x) \nabla V(x) + \lambda_i G(x) \nabla h_i(x)$. 
Differentiating yields 
$\partial_k f_{cl}(x) = j_{k} - (\partial_k \lambda_0) G \nabla V + (\partial_k \lambda_i) G \nabla h_i$.
At the boundary equilibrium point $x_e \in \mathcal{E}_{\partial \mathcal{C}_i}$, $\lambda_0 = p \gamma(V(x_e))$ and 
$\partial_k f_{cl}(x_e)$ can be written as
\begin{align}
\partial_k f_{cl}(x_e) \!=\! [J_i(x_e, \lambda_i)]_k 
&\!-\! \partial_k \lambda_0(x_e) G \nabla V \nonumber \\
&\!+\! \partial_k \lambda_i(x_e) G \nabla h_i
\label{eq:jacobian_column}
\end{align}
Substituting the expressions for $\partial_k \lambda_0(x_e)$ and $\partial_k \lambda_i(x_e)$ obtained from \eqref{eq:partial_lambda2} into \eqref{eq:jacobian_column}
yields an involved expression that can be greatly simplified by using the definitions of $z_1$, $z_2$ in \dref{def:z_vectors}, $\eta$ in \eqref{eq:eta} and property (iv) of $\eta$.
After simplifications, the resulting expression for the $k$-th column of the closed-loop Jacobian at the boundary equilibrium point $x_e \in \mathcal{E}_{\partial \mathcal{C}_i}$ is
\begin{align}
\partial_k f_{cl}(x_e) &\!=\! 
\left( \underbrace{ I_{n} \!-\! G z_1 z_1^\mathsf{T} \!-\! p \eta G z_2 z_2^\mathsf{T} }_{ I_n - G Z N_1 Z^\mathsf{T} } \right) \! [J_i(x_e, \lambda_i)]_k 
\label{eq:jacobian_column2} \\
& \qquad \qquad \qquad \qquad \qquad \quad - G [ Z N_1 \Psi Z^\mathsf{T} ]_k \nonumber
\end{align}
Then, letting $k \in \{1, \cdots, n\}$ and combining the $n$ partial derivatives as column vectors in \eqref{eq:jacobian_column2} to form the closed-loop Jacobian matrix $J_{cl}(x_e)$ yields \eqref{eq:jacobian_boundary}. \myqed
\end{pf}
\vspace{-4mm}
{
The computation of the closed expression for the closed-loop Jacobian \eqref{eq:jacobian_boundary} at a boundary equilibrium point $x_e \in \mathcal{E}_{\partial \mathcal{C}_i}$ is the first step towards determining its stability properties. Notice that it has a very particular structure, containing the Jacobian \eqref{eq:Jfi} and also matrices that closely resemble oblique projection matrix operators. The properties of these matrices and the following secondary lemma will be essential for deriving the main result of this section.
}
\begin{lem}
\label{lemma:basis} { Let $G(x)$ be positive semi-definite.}
Assume $L_g h_i(x) \ne 0$, and consider orthogonality with respect to the inner product $\inner{\cdot}{\cdot}{G(x)}$.
\begin{enumerate}[label=(\roman*)]
\item If $L_g V(x) \neq 0$, define the set $\mathcal{Z}$ as $\mathcal{Z} = \{ z_1, z_2 \}$.
\item If $L_g V(x) = 0$, define the set $\mathcal{Z}$ as $\mathcal{Z} = \{ z_1 \}$.
\end{enumerate}
Let the set $\mathcal{W} = \{ w_1, \cdots, w_{\dim{ \mathcal{W} }} \}$ be an orthonormal basis for 
$\Span{ \mathcal{Z} }^{\perp} \subset \mathbb{R}^{n}$, that is, $\inner{w_i}{w_j}{G} = \delta_{ij}$ for 
$i,j \in \{1, \cdots, \dim{\mathcal{W}} \}$ (since $\mathcal{W}$ is an orthonormal set) and 
$\inner{w_i}{z_k}{G} = 0$ for $i \in \{1, \cdots, \dim{ \mathcal{W} } \}$, $k \in \{1, 2\}$.
Then, in both cases, the set $\mathcal{B} = \mathcal{Z} \cup \mathcal{W}$ is an orthogonal basis for $\mathbb{R}^n$.
\end{lem}
\vspace{-5mm}
\begin{pf}
First notice that $\mathcal{B} = \mathcal{Z} \cup \mathcal{W}$ is an orthogonal set in both cases,
since $\inner{z_1}{z_2}{G(x)} = 0$ and $\inner{w_i}{z_k}{G(x)} = 0$, $i \in \{1, \cdots, \dim{\mathcal{W}} \}$, $k \in \{1, 2\}$ by construction.
To prove that it is also a basis for $\mathbb{R}^n$, let the following be the linear independence equations for the vectors in $\mathcal{B}$ for $L_g V(x) \neq 0$ and $L_g V(x) = 0$, respectively:
\begin{align}
\beta_1 z_1 + \beta_2 z_2 + \sum^{n-2}_{i=1} \beta_{i+2} w_i = 0 \,, \quad L_g V(x) \neq 0
\label{eq:linear_independence1} \\
\beta_1 z_1 + \sum^{n-1}_{i=1} \beta_{i+1} w_i = 0 \,, \quad L_g V(x) = 0
\label{eq:linear_independence2}
\end{align}
Taking the inner product of \eqref{eq:linear_independence1}-\eqref{eq:linear_independence2} with $z_1$ yields $\beta_1 = 0$ in both cases, 
{ since $\inner{z_1}{z_1}{G} \ne 0$, $\inner{z_1}{z_2}{G} = 0$} and $\inner{w_i}{z_1}{G} = 0 \,, \forall i \in \{1, \cdots, \dim{\mathcal{W}}\}$.
Then, two different cases must be considered:\\ 
{\bf Case (i)}: $L_g V(x) \neq 0$, $\dim{\mathcal{W}} = n-2$. { In this case, $G z_2 \ne 0$}. Taking the inner product of \eqref{eq:linear_independence1} 
with $z_2$ yields $\beta_2 = 0$, since $\inner{z_1}{z_2}{G(x)} = 0$ and $\inner{w_i}{z_2}{G(x)} = 0$, $i \in \{1, \cdots, n-2\}$.
Taking the inner product of \eqref{eq:linear_independence1} with $w_j$, $j \in \{1,\cdots, n-2\}$ yields 
$\beta_3 = \cdots = \beta_{n} = 0$, since $\inner{z_1}{w_j}{G(x)} = \inner{z_2}{w_j}{G(x)} = 0$ and $\inner{w_i}{w_j}{G(x)} = \delta_{ij}$.\\
{\bf Case (ii)}: $L_g V(x) = 0$, $\dim{\mathcal{W}} = n-1$. { In this case, $G z_2 = 0$} and $z_2$ is not contained in $\mathcal{Z}$.
The absence of $z_2$ in $\mathcal{Z}$ is compensated by the presence of an extra basis vector for $\mathcal{W}$ appearing in the 
summation (since $\dim{\mathcal{W}} = n-1$ in this case).
Taking the inner product of \eqref{eq:linear_independence2} with $w_j$, $j \in \{1,\cdots, n-1\}$ yields 
$\beta_2 = \cdots = \beta_{n} = 0$, since $\inner{z_1}{w_j}{G(x)} = 0$ and $\inner{w_i}{w_j}{G(x)} = \delta_{ij}$.\\
Therefore, $\beta_1 = \beta_2 = \cdots = \beta_n = 0$ in both cases, and therefore the set $\mathcal{B} = \mathcal{Z} \cup \mathcal{W}$ 
forms a basis for $\mathbb{R}^n$.\myqed
\end{pf}
\vspace{-4mm}
{
Next, we present the main result of this section.}
\begin{thm}[Stability of Boundary Equilibria]
\label{thm:curvature}
Under Assumption \ref{assumption:disjoint_barriers}, consider a boundary equilibrium point 
$x_e \in \mathcal{E}_{\partial \mathcal{C}_i}$ of the closed-loop system \eqref{eq:feedback_system} 
with controller \eqref{eq:QP_control} such that $L_g h_i(x_e) \ne 0$, with corresponding $i$-th KKT 
multiplier given by $\lambda_e \ge 0$.
If there exists $v \in \{\nabla h_i(x_e)\}^\perp$ such that
\begin{align}
v^\mathsf{T} J_{f_i}(x_e, \lambda_e) v > 0 \,,
\label{eq:curvature_condition}
\end{align}
then $x_e$ is {\it unstable}. Otherwise, $x_e$ is {\it stable}. 
In particular, if $v^\mathsf{T} J_{f_i}(x_e, \lambda_e) v < 0 \,\, \forall v \in \{\nabla h_i(x_e)\}^\perp$, 
then $x_e$ is asymptotically stable.
The matrix function $J_{fi}$ is the Jacobian of the vector field $f_i$, as defined in \eqref{eq:Jfi}.
\end{thm}
\vspace{-4mm}
\begin{pf}
\label{sec:proof_curvature}
Consider a boundary equilibrium point $x_e \in \mathcal{E}_{\partial \mathcal{C}_i}$ with $L_g h_i(x_e) \ne 0$.
%
%
The corresponding Lyapunov equation for $x_e$ is then given by
%
%
\begin{align}
Y &= J_{cl}(x_e)^\mathsf{T} X + X J_{cl}(x_e)
\label{eq:lyap_equation} 
\end{align}
where $J_{cl}(x_e)$ is expressed by \eqref{eq:jacobian_boundary}.
Define
\begin{align}
X &= Z \Lambda_z Z^\mathsf{T} + W \Lambda_w W^\mathsf{T} > 0
\label{eq:X}
\end{align}
where the columns of matrices $Z$ and $W$ are given by the basis vectors of $\mathcal{Z}$ and $\mathcal{W}$ of \lref{lemma:basis}.
Matrices $\Lambda_z$, $\Lambda_w$ are diagonal and positive definite, with dimensions combatible to the dimensions of 
$\mathcal{Z}$ and $\mathcal{W}$ from \lref{lemma:basis}, that is,
(i) if $L_g V(x_e) \neq 0$, $\dim{ \mathcal{Z} } = 2$, $\dim{ \mathcal{W} } = n-2$, 
(ii) if $L_g V(x_e) = 0$, $\dim{ \mathcal{Z} } = 1$, $\dim{ \mathcal{W} } = n-1$.
From the properties of vectors \eqref{eq:z1}-\eqref{eq:z2} and of the subspace $\mathcal{W}$, we have 
(i) $Z^\mathsf{T} G Z = diag\{ 1, p^{-1}(\eta^{-1} - 1) \}$, 
(ii) $Z^\mathsf{T} G W = 0$ and 
(iii) $W^\mathsf{T} G W = I_{n-2}$ if $L_g V(x_e) \neq 0$ and $W^\mathsf{T} G W = I_{n-1}$ if $L_g V(x_e) = 0$.
Substituting the closed-system Jacobian \eqref{eq:jacobian_boundary} and \eqref{eq:X} in \eqref{eq:lyap_equation} and using 
the properties (i), (ii) and (iii) for matrices $Z$ and $W$, it is possible to write $Y$ in the following way, no matter the 
case considered ($L_g V(x_e)\neq 0$ or $L_g V(x_e) = 0$):
\begin{align}
Y &= J_i^\mathsf{T} \overline{X} + \overline{X} J_i - Z \Omega Z^\mathsf{T} 
\label{eq:lyap_equation2}
\end{align}
where $\Omega =\Omega^\mathsf{T} \ge 0$. 
The expressions for $\overline{X}$, $\Omega$ and $Z$ depend on the considered case:\\
{\bf Case (i)}: if $L_g V(x_e) \neq 0$, $\dim{ \mathcal{Z} } = 2$, $Z = [\, z_1 \,\,\, z_2 \,] \in \mathbb{R}^{n \times 2}$,
$\Lambda_{z} = \Diag{\lambda_{z_1}, \lambda_{z_2}}$. 
In \eqref{eq:lyap_equation2}, $\overline{X} = \eta \lambda_{z_2} z_2 z_2^\mathsf{T} + W \Lambda_w W^\mathsf{T} \ge 0$
and $\Omega = \Psi^\mathsf{T} N \Lambda_z + \Lambda_z N \Psi$, with $N = \Diag{1,1-\eta}$.
Here, $W = [\, w_{1} \,\, \cdots \,\, w_{n-2} \,] \in \mathbb{R}^{n \times (n-2)}$.\\
{\bf Case (ii)}: if $L_g V(x_e) = 0$, $\dim{ \mathcal{Z} } = 1$, $Z = z_1 \in \mathbb{R}^{n}$,
$\Lambda_{z} = \lambda_{z_1} \in \mathbb{R}_+$.
In \eqref{eq:lyap_equation2}, $\overline{X} = W \Lambda_w W^\mathsf{T} \ge 0$
and $\Omega = 2 \lambda_{z_1} \alpha^{\prime}(h_i)$.
Here, $W = [\, w_{1} \,\, \cdots \,\, w_{n-1} \,] \in \mathbb{R}^{n \times (n-1)}$.\\

\vspace{-5mm}
In both cases, $\overline{X}$ has no term $z_1 z_1^\mathsf{T}$.
By Chetaev's instability theorem, $x_e$ is unstable if there exists $v \in \mathbb{R}^{n}$ such that $v^\mathsf{T} Y v > 0$ in \eqref{eq:lyap_equation2} \cite{Khalil2002}.
Then, the quadratic form $v^\mathsf{T} Y v$ yields
\begin{align}
v^\mathsf{T} Y v = 2 v^\mathsf{T} \overline{X} J_i v - v^\mathsf{T} Z \Omega Z^\mathsf{T} v
\label{eq:yYy}
\end{align}
Let $v \in \{ \nabla h_i(x_e)\}^{\perp}$.
Then:\\
{\bf Case (i)}: if $L_g V(x_e) \neq 0$, the second term on the right-hand side of \eqref{eq:yYy} becomes
\begin{align}
v^\mathsf{T} Z \Omega Z^\mathsf{T} v 
&= v^\mathsf{T} \underbrace{ \left( \gamma^{\prime}(V) \sigma_{z_2} (1 - \eta) z_2 z_2^\mathsf{T} \right) }_{ p \gamma^{\prime}(V) \overline{X} G z_2 z_2^\mathsf{T} } v
\label{eq:second_term} \\
&= v^\mathsf{T} \left( p \gamma^{\prime}(V) \overline{X} G \nabla V \nabla V^\mathsf{T} \right) v
\end{align}
where we have used the fact that $\inner{w_i}{z_2}{G} = 0$ and 
property (iii) of \eqref{eq:eta}.
Then, \eqref{eq:yYy} can be rewritten as
\begin{align}
v^\mathsf{T} Y v = 2 v^\mathsf{T} \overline{X} 
\underbrace{ \left( J_i - p \gamma^{\prime}(V) G \nabla V \nabla V^\mathsf{T} \right) }_{J_{f_i}(x_e, \lambda)} v
\label{eq:yYy2}
\end{align}
{\bf Case (ii)}: if $L_g V(x_e) = 0$, the second term on the right-hand side of \eqref{eq:yYy} vanishes, since $Z = z_1$ and therefore, $Z^\mathsf{T} v = 0$.
Then, { using the fact that $G z_2 = G \nabla V = 0$ in this case}, \eqref{eq:yYy} yields
\begin{align}
v^\mathsf{T} Y v = 2 v^\mathsf{T} \overline{X} J_i v 
= 2 v^\mathsf{T} \overline{X} \left( J_{f_i}(x_e, \lambda)\big|_{G \nabla V = 0} \right) v
\label{eq:yYy3}
\end{align}
Therefore, { both cases $L_g V(x_e) \neq 0$ and $L_g V(x_e) = 0$} result in $v^\mathsf{T} Y v = 2 v^\mathsf{T} \overline{X} J_{f_i}(x_e, \lambda) v$.
Given any matrix $M$, since $\overline{X}$ is symmetric and positive semi-definite, matrices $\overline{X} M$ and $\overline{X}^{\frac{1}{2}} M \overline{X}^{\frac{1}{2}}$ share the same nonzero eigenvalues. Since the spectra of $\overline{X}^{\frac{1}{2}} M \overline{X}^{\frac{1}{2}}$ is real, so is the spectra of $\overline{X} M$.
Since $\overline{X} \ge 0$ is arbitrary, with its one-dimensional nullspace spanned by $G \, z_1 \ne 0$, it is always possible to choose $\Lambda_z$ and $\Lambda_w$ in such a way that $v \in \{ \nabla h_i(x_e)\}^{\perp}$ is an eigenvector of $\overline{X}$ with a corresponding strictly positive eigenvalue $\lambda(\overline{X}) > 0$. Then, $v^\mathsf{T} Y v$ yields
\begin{align}
v^\mathsf{T} Y v = 2 \lambda(\overline{X}) v^\mathsf{T} J_{f_i}(x_e, \lambda) v
\label{eq:yYy4}
\end{align}
Then, $x_e \in \mathcal{E}_{\partial \mathcal{C}_i}$ is unstable if the right-hand side of \eqref{eq:yYy4} is strictly positive, demonstrating \eqref{eq:curvature_condition}.\\
To show that $x_e \in \mathcal{E}_{\partial \mathcal{C}_i}$ is locally stable otherwise, we proceed as follows.
The first order Taylor series approximation of the closed-loop system on a neighborhood of $x_e$ is
$\dot{x} = J_{cl}(x_e) \Delta x$ with $\Delta x = (x - x_e)$ being a disturbance vector around the equilibrium point. Let us write this disturbance vector using the basis $\{ z_1(x_e), v_1, \cdots, v_{n-1}(x_e) \}$, where the $v_1, \cdots, v_{n-1}$ are fixed basis vectors for 
$\{ \nabla h_i(x_e) \}^\perp$. Therefore, $\Delta x = \beta z_1(x_e) + v$, with 
$v = \sum^{n-1}_{i=1} \beta_i v_i$. Note that $v^\mathsf{T} \nabla h_i(x_e) = 0$ by construction.
Here, $\beta, \beta_1, \cdots, \beta_{n-1}$ represent the coordinates of $\Delta x$ in the new basis.
Computing the inner product $\inner{z_1(x_e)}{\dot{x}}{}$ yields
\begin{align}
\inner{z_1(x_e)}{\dot{x}}{} &= 
\inner{z_1}{\dot{\beta} z_1 + \sum^{n}_{i=1} \dot{\beta_i} v_i}{} = \dot{\beta} \norm{z_1(x_e)}{}^2 \nonumber \\
&= z_1^\mathsf{T} J_{cl}(x_e) (\beta z_1(x_e) + v)
\label{eq:inner_product}
\end{align}
Since $z_1^\mathsf{T} J_{cl}(x_e) = -\alpha^{\prime}(h_i(x_e)) z_1(x_e)^\mathsf{T}$ in \eqref{eq:inner_product}, the dynamics of $\beta$ is given by
$\dot{\beta} = -\alpha^{\prime}(h_i(x_e)) \beta$.
Since $\alpha^{\prime}$ is a $\mathcal{K}_\infty$ function, $\alpha^{\prime}(h_i(x_e)) > 0$, which means that $\beta \rightarrow 0$. Replacing the dynamics of $\beta$ into the Taylor expansion yields the following dynamics for $v$:
\begin{align}
\dot{v} &= \beta \left(J_{cl}(x_e) + \alpha^{\prime}(h_i) I_n \right) z_1 + J_{cl}(x_e) v
\label{eq:v_dynamics}
\end{align}
%
%
Define the Lyapunov candidate $V(\beta, v) = \frac{1}{2} \beta^2 + v^\mathsf{T} X v > 0$, with $X$ given by \eqref{eq:X}. Taking its time derivative and using the dynamics of $\beta$ and \eqref{eq:v_dynamics} yields
\begin{align}
\dot{V} &= -\alpha^{\prime}(h_i) \beta^2 + v^\mathsf{T} Y v  \label{eq:dotV_stability} \\
&+ 2 v^\mathsf{T} X \left( J_{cl}(x_e) + \alpha^{\prime}(h_i) I_n \right) z_1(x_e) \beta \nonumber
\end{align}
Equation \eqref{eq:dotV_stability} shows that, since the dynamics of $\beta$ is decoupled and asymptotically stable, the sign of $\dot{V}$ is eventually determined by the term 
$v^\mathsf{T} Y v$. 
By \eqref{eq:yYy4}, if $v^\mathsf{T} J_{f_i}(x_e, \lambda) v \le 0 \,\, \forall v \in \{\nabla h_i(x_e)\}^{\perp}$, then $v^\mathsf{T} Y v \le 0$, and $x_e$ is stable.
%
In particular, if $v^\mathsf{T} J_{f_i}(x_e, \lambda) v < 0 \,\, \forall v \in \{\nabla h_i(x_e)\}^{\perp}$,
then $x_e$ is asymptotically stable. \myqed 
\end{pf}
\vspace{-4mm}
{
Through condition \eqref{eq:curvature_condition}, \tref{thm:curvature} provides a test for determining the stability of a boundary equilibrium point $x_e \in \mathcal{E}_{\partial \mathcal{C}_i}$: its stability properties are dependent on the orthogonal subspace of the CBF gradient and on the Jacobian of the function $f_i$ as defined in \eqref{eq:Jfi} at $x_e$. The dependency of the stability on $\{\nabla h_i(x_e)\}^{\perp}$ can be seen by noticing that $z_1$ is a left eigenvector of the closed-loop Jacobian $J_{cl}(x_e)$ at \eqref{eq:jacobian_boundary}, with corresponding eigenvalue $-\alpha^{\prime}(0) < 0$, which is strictly negative since $\alpha$ is a class $\mathcal{K}_{\infty}$ function. Therefore, $\nabla h_i(x_e)$ is tangent to a stable manifold of $x_e$, and therefore its stability must be determined by the remaining eigenvalues associated to 
$\{ \nabla h_i(x_e) \}^\perp$.
Furthermore, a corollary of \tref{thm:curvature} was previously presented at \cite{Reis_LCSS} for a simple integrator in $\mathbb{R}^n$: in this case, $f(x) = 0$ and $G(x) = I_n$, and 
$J_{f_i}(x,\lambda) = \lambda H_{h_i}(x) - p H_{\overline{V}}(x)$.
Therefore, condition \eqref{eq:curvature_condition} reduces to a strictly positive difference between the {\it weighted CBF and CLF curvatures} at $x_e$.}
\section{CLF Compatibility}
\label{sec:clf_compatibility}

In light of \tref{thm:curvature}, the stability properties of boundary equilibrium points with $L_g h(x) \neq 0$ are determined 
by the Jacobian $J_{f_i}$ in \eqref{eq:Jfi}, which depends on the system dynamics and on the geometry of the CLF and CBFs.
Since the system dynamics must be kept general and the CBFs must provide a model for the safety requirements of the problem, 
we consider the following question: is it possible to find a valid CLF such that all boundary equilibrium points that are 
unremovable by the technique of \cite{tan2024undesired} are either removed or unstable?
Our objective in this section is to prove that this is indeed possible by providing a method for computing such CLF, considering
a particular type of system and CLF-CBF geometry.

\begin{defn}[CLF Compatibility]
\label{def:compatibility}
Under Assumptions \ref{assumption:initial_state}-\ref{assumption:clf_condition}, a CLF with global minimum in $x_0 \in \mathbb{R}^n$ 
is said to be {\it $i$-th compatible} if $x_0$ is the only stable equilibrium point of the closed-loop system \eqref{eq:feedback_system} 
with only the $i$-th CBF implemented in the 
QP controller \eqref{eq:QP_control}.
\end{defn}
%
{
By \dref{def:compatibility}, if the CLF is compatible with the $i$-th CBF, even when undesirable boundary equilibrium points are present (since they are all unstable), it implies that the trajectories fail to converge to the CLF global minimum $x_0$ only in a set of measure zero containing any existing boundary equilibrium points, provided that no other types of undesired attractors exist, such as limit cycles.
}
\subsection{Quadratic CLF Compatibility}
\label{sec:prob_formulation}
Consider the following classes of systems:
\begin{enumerate}
    \item[(i)] Nonlinear systems of the form $\dot{x} = g(x) u$ with full-rank $g(x)$ for all $x \in \mathbb{R}^n$.
    \item[(ii)] Linear Time-Invariant systems $\dot{z} = A z + B u$, $A \in \mathbb{R}^{n \times n}$, $B \in \mathbb{R}^{n \times m}$ (LTI), where $z = x - x_0$.
\end{enumerate}
Let the transformed CLF $\overline{V}$ (see \dref{eq:transCLF}) and the $N$ CBFs be parametrized by quadratic polynomials on $\mathbb{R}^n$:
\begin{align}
\overline{V}(x) &= \frac{1}{2} \Delta x^\mathsf{T} H_{\overline{V}} \Delta x \ge 0 \,, &\Delta x = x - x_0
\label{eq:parametric_CLF} \\
h_i(x) &= \frac{1}{2} \left( \Delta x_i^\mathsf{T} H_{h_i} \Delta x_i - 1 \right) \,, &\Delta x_i = x - c_i
\label{eq:parametric_CBF} 
\end{align}
where the parameters of the CLF and $i$-th CBF are: the constant, positive-definite Hessian matrices $H_{\overline{V}}, H_{h_i} \in \mathcal{S}^n_+$ 
determining the elliptical shapes of their level sets and their centers $x_0, c_i \in \mathbb{R}^n$, $i \in \{0, 1, \cdots, N\}$.
%
Recall that due to property (ii) of \pref{prop:properties_transformed_CLF}, $V$ can be computed from $\overline{V}$ as defined in \eqref{eq:parametric_CLF}, and then used in the QP controller 
\eqref{eq:QP_control}.

Using the gradients of \eqref{eq:parametric_CLF}-\eqref{eq:parametric_CBF}, consider the expressions for $f_i(x,\lambda)$ for each class of systems after performing a state translation
$\nu = \Delta x_i$ (that is, $x = \nu + c_i$):\\
{\bf Case (i)} For nonlinear systems of the form $\dot{x} = g(x) u$ with full-rank $g(x)$, $f_i(\nu,\lambda) = G(\nu) ( P(\lambda) \nu - w )$, where $G(x) > 0$ and
$P(\lambda) = \lambda M - N$ is a symmetric Linear Matrix Pencil (LMP) \cite{Ikramov1993} with $M = H_{h_i}$, $N = p H_{\overline{V}}$, and $w = N (c_i - x_0)$.\\
{\bf Case (ii)} For LTI systems, $f_i(\nu,\lambda) = P(\lambda) \nu - w$, where $P(\lambda) = \lambda M - N$ is an LMP 
with $M = B B^\mathsf{T} H_{h_i}$, $N = p B B^\mathsf{T} H_{\overline{V}} - A$, and $w = N (c_i - x_0)$.\\

In both cases, the equilibrium points in $\partial \mathcal{C}_i$ must satisfy $f_i(\nu, \lambda) = 0$, that is
\begin{align}
P(\lambda) \nu &= w \,, \quad \lambda \ge 0
\label{eq:quadratic_equilibrium}          
\\
\nu^\mathsf{T} H_{h_i} \nu &= 1
\label{eq:quadratic_boundary} 
\end{align}
%
%
%
%
{where $P: \mathbb{R} \rightarrow \mathbb{R}^{n \times n}$ is an LMP.
Equation \eqref{eq:quadratic_equilibrium} is equivalent to $f_i(\nu, \lambda) = 0$, while \eqref{eq:quadratic_boundary} describes the $i$-th CBF boundary in terms of the translated state $\nu$.
If $P$ is regular, $P(\lambda)^{-1}$ exists $\forall \lambda \not\in \sigma_P$, and \eqref{eq:quadratic_equilibrium} can be solved for $\nu$, yielding
\begin{align}
\nu(\lambda) = P(\lambda)^{-1} w \,, \quad \lambda \in \mathbb{R}_+ \setminus \sigma_P
\label{eq:nu_sol}
\end{align}
While \eqref{eq:quadratic_equilibrium} describes the equilibrium manifold where boundary equilibrium points occur {\it implicitly}, \eqref{eq:nu_sol} describes the same manifold {\it explicitly}, with $\nu$ as an explicit function of $\lambda$.
Since $\lambda \in \mathbb{R}_+$, $v(\lambda)$ in \eqref{eq:nu_sol} is not defined only at the real, generalized eigenvalues of $P$.
\begin{defn}[Q-Function]
Define the scalar function $q(\lambda) = \nu(\lambda)^\mathsf{T} H_{h_i} \nu(\lambda)$, where $\nu(\lambda)$ is given by \eqref{eq:nu_sol}:
%
\begin{align}
q(\lambda) &= 
\underbrace{ w^\mathsf{T} P(\lambda)^{-\mathsf{T}} }_{\nu(\lambda)^\mathsf{T}} H_{h_i} \underbrace{ P(\lambda)^{-1} w }_{\nu(\lambda)} = \frac{n(\lambda)}{|P(\lambda)|^2}
\label{eq:qfunction}
\end{align}
where $n(\lambda) = w^\mathsf{T} \Adj{P(\lambda)}^{\mathsf{T}} H_{h_i} \Adj{P(\lambda)} w$ and $|P(\lambda)|^2$ are non-negative polynomials with real coefficients.
\end{defn}
}
Equation \eqref{eq:qfunction} defines the Q-function for the $i$-th barrier $h_i$. 
Since $H_{h_i}$ is positive-definite and since $P(\lambda)^{-1}$ must be full-rank for all $\lambda \in \mathbb{R} \setminus \sigma_P$, we conclude that $q$ is a positive definite function.
As will be demonstrated, for the classes of systems (i) and (ii) described in the beginning of this section, it encodes all the necessary information to compute the equilibrium points on the $i$-th boundary.
Due to \eqref{eq:quadratic_boundary}, every $\lambda_e \ge 0\,, \lambda_e \notin \sigma(P)$ satisfying $q(\lambda_e) = 1$ corresponds to a boundary equilibrium point $x_e = \nu(\lambda_e) + c_i = P(\lambda_e)^{-1} w + c_i \in \partial \mathcal{C}_i$.
\begin{thm}[Q-Function Properties]
\label{thm:q_properties}
Consider the safety-critical control problem described in \sref{sec:prob_formulation}, under Assumptions \ref{assumption:initial_state}-\ref{assumption:clf_condition}, 
and the Q-function $q(\lambda)$ associated to the $i$-th CBF. Then,
\begin{enumerate}[label=(\roman*)]
\item If \aref{assumption:initial_state} holds, then $q(0) \ge 1$.
\item For any $\lambda_i \in \sigma_P$ that is not also a root of the numerator polynomial $n(\lambda)$, 
$\lim_{\lambda \rightarrow \lambda_i} q(\lambda) = \infty$.
\item If $q(\lambda)$ is proper, the closed-loop system \eqref{eq:feedback_system} has at least one boundary equilibrium point. 
\item Let $R(\lambda) \in \mathbb{R}^{n \times(n-1)}$ be a polynomial matrix on the orthogonal space of $H_{h_i} \Adj{P(\lambda)} w \in \mathbb{R}^{n}$, 
that is, satisfying $R(\lambda)^\mathsf{T} H_{h_i} \Adj{P(\lambda)} w = 0 \,\, \forall \lambda \in \mathbb{R}$.
Define the {\it stability} polynomial matrix as
\begin{align}
S(\lambda) = R(\lambda)^\mathsf{T} \left( P(\lambda) + P(\lambda)^\mathsf{T} \right) R(\lambda) \,.
\label{eq:stability_polynomial}
\end{align}
Then, an equilibrium point $x_e = \nu(\lambda_e) + x_i \in \partial \mathcal{C}_i$ is stable if and only if $S(\lambda_e) \le 0$. Otherwise, it is unstable.
\item The maximum number of negative semi-definite intervals of $S(\lambda)$ is $n$.
\end{enumerate}
\end{thm}
\vspace{-4mm}
\begin{pf}
For $P(\lambda)$ defined in both Cases (i) and (ii), evaluating \eqref{eq:quadratic_equilibrium} at $\lambda = 0$ yields 
$P(0) \nu = w$, yielding $\nu(0) = x_0 - c_i$. Therefore, $q(0) = (x_0 - c_i)^\mathsf{T} H_{h_i} (x_0 - c_i)$. Then, by \eqref{eq:parametric_CBF},
$q(0) > 1$ is equivalent to $h_i(x_0) > 0$, which means that $x_0 \in \mathcal{C}_i$, that is, \aref{assumption:initial_state} is satisfied. 
This proves (i).

Since $q$ is a positive-definite rational function with denominator given by the polynomial 
$|P(\lambda)|^2$, every $\lambda_i \ge 0$ such that $\sigma_P$ satisfies $|P(\lambda)| = 0$, ($\lambda_i$ is a pole of $q$).
Therefore, if $\lambda_i \in \sigma_P$ is not also a zero of $q$ (that is, a root of the numerator polynomial $n(\lambda)$), then $\lim_{\lambda \rightarrow \lambda_i} \frac{n(\lambda)}{|P(\lambda)|^2} = \infty$.
This proves (ii).

Consider an arbitrary closed interval $\mathcal{I} \subset \mathbb{R}_+$. If $q(\lambda) > 1$ for all $\lambda \in \mathcal{I}$, 
then $\mathcal{I}$ does not contain equilibrium point solutions. Using \eqref{eq:qfunction}, 
$q(\lambda) > 1$  over $\mathcal{I}$ implies $n(\lambda) - |P(\lambda)|^2 > 0$ over $\mathcal{I}$.
%
%
If it were possible to guarantee this condition for the entire positive real line $\mathcal{I} = \mathbb{R}_+$, then no boundary equilibrium points would exist. However, this is impossible in general, since $q(\lambda) \ge 0$ and $\lim_{\lambda \rightarrow +\infty} q(\lambda) = 0$ in case $q(\lambda)$ is proper, proving (iii).
%
For the considered problem, this also demonstrates the impossibility of removing all undesirable equilibrium points for certain types of systems.

Next, consider the translated state $\nu = x - c_i$ and the CL and LTI systems in \sref{sec:prob_formulation}. 
Then, the boundary equilibrium point is $\nu_e = \nu(\lambda_e) \in \partial \mathcal{C}_i$ for some $\lambda_e \ge 0$ such that $q(\lambda_e) = 1$, and $\nu$ as defined in \eqref{eq:nu_sol}.
\\
{\bf Case (i)} For CL systems, the columns of $J_{f_i}(\nu,\lambda)$ are given by
$\partial_k f_i(\nu,\lambda_e) = \partial_k G(\nu) ( P(\lambda) \nu - w ) + G(\nu) [P(\lambda)]_k$.
Using \eqref{eq:quadratic_equilibrium}, $J_{f_i}(\nu_e,\lambda_e) = G(\nu_e) P(\lambda_e)$.
\\
{\bf Case (ii)} For LTI systems, the columns of $J_{f_i}(\nu,\lambda)$ are given by
$\partial_k f_i(\nu,\lambda) = [P(\lambda)]_k$.
Using \eqref{eq:quadratic_equilibrium}, $J_{f_i}(\nu_e,\lambda) = P(\lambda_e)$.\\
Since $\nabla h_i(\nu) = H_{h_i} \nu$, using \eqref{eq:nu_sol}, the vector rational function
$\nabla h_i(\lambda) = H_{h_i} P(\lambda)^{-1} w$ describes the barrier gradient 
in the equilibrium manifold consisting of states $\nu$ such that $f_i(\nu, \lambda) = 0$.
At the equilibrium point $\nu_e$, we have $\nabla h_i(\nu_e) = \nabla h_i(\lambda_e)$.
From \tref{thm:curvature}, for any of the two cases, if there exists $v \in \{\nabla h_i(\lambda_e)\}^\perp$ such that
$v^\mathsf{T} P(\lambda_e) v > 0$, then $\nu_e$ is unstable.
Let $v = \projection{\nabla h_i}(\lambda_e) z$ with an arbitrary $z \in \mathbb{R}^n$, where 
$\projection{\nabla h_i}(\lambda) = \norm{\nabla h_i(\lambda)}{}^2 I_n - \nabla h_i(\lambda) \nabla h_i(\lambda)^\mathsf{T}$ is a scaled projection matrix. Then, $v$ is a projection into $\{ \nabla h_i(\nu_e) \}^{\perp}$ for any $z \in \mathbb{R}^n$. 
Substituting $v$ into \eqref{eq:curvature_condition}, 
$\nu_e$ is unstable if $\exists z \in \mathbb{R}^n$ such that
\begin{align}
z^\mathsf{T} S_{null}(\lambda_e) z > 0 \,,
\label{eq:lift_curvature_condition}
\end{align}
where $S_{null}(\lambda) = \projection{\nabla h_i}(\lambda)^\mathsf{T} ( P(\lambda) + P(\lambda)^\mathsf{T} ) \projection{\nabla h_i}(\lambda)$.
By construction, $\nabla h_i(\lambda)$ is in the null-space of $S_{null}(\lambda)$ for all $\lambda \in \mathbb{R}$.

The null-space of $\nabla h_i(\lambda)$ is the same as that of the vector polynomial $H_{h_i} \Adj{P(\lambda)} w$, 
since they differ only by a scalar factor of $|P(\lambda)|^{-1}$.
Since $|P(\lambda)|$ is of maximum degree $n$ and $P(\lambda) \Adj{P(\lambda)} = |P(\lambda)| I_n$, $\Adj{P(\lambda)}$ is a polynomial matrix of maximum degree $n-1$. 
Let $H_{h_i} \Adj{P(\lambda)} w = v_0 + \lambda v_1 + \cdots + \lambda^{l} v_{l}$, $l \le n-1$, with $v_i$ being constant vector coefficients. 
Let $r(\lambda) = r_0 + \lambda r_1 + \cdots + \lambda^{d} r_{d} \in \mathbb{R}^n$ be a vector polynomial of degree $d \in \mathbb{N}$ 
in the null-space of $\nabla h_i(\lambda)$. Since $r(\lambda)^\mathsf{T} ( v_0 + \lambda v_1 + \cdots + \lambda^{l} v_{l} ) = 0 \,\, \forall \lambda \in \mathbb{R}$,
their coefficients must satisfy $\sum^k_{i=0} v_i^\mathsf{T} r_{i-k} = 0$ for $k = \{ 0, l+d \}$.
These $l+d+1$ equations can be stacked in matrix form $V \overline{r} = 0$, with $\overline{r}^\mathsf{T} = [\, r_0^\mathsf{T} \,\, r_1^\mathsf{T} \,\, \cdots \,\, r_d^\mathsf{T} \,] \in \mathbb{R}^{(d+1)n}$ 
, $V \in \mathbb{R}^{(l+d+1)\times(d+1)n}$ \cite{Gantmacher1980}[Chapter XII, Section 3].
%
%
%
%
Since $\mathbb{R}^n = \Span{ H_{h_i} \Adj{P(\lambda)} w } \bigoplus \{ H_{h_i} \Adj{P(\lambda)} w \}^{\perp} \,, \forall \lambda \in \mathbb{R}$ 
(a direct sum of orthogonal subspaces), one can always obtain $n-1$ linearly independent basis vectors for $\{ H_{h_i} \Adj{P(\lambda)} w \}^{\perp}$ from the 
vector coefficients drawn from a basis of $\mathcal{N}(V)$.

Then, the arbitrary vector $z$ from \eqref{eq:lift_curvature_condition} can be decomposed using the basis 
\begin{align}
\{ \nabla h_i(\lambda_e), r_1(\lambda_e), \cdots, r_{n-1}(\lambda_e) \}
\label{eq:orthogonal_basis}
\end{align}
where $r_i(\lambda)$, $i \in \{ 1, \cdots, n-1\}$ are basis polynomials for $\{ H_{h_i} \Adj{P(\lambda)} w \}^{\perp}$.
Then, $z = \beta \nabla h_i(\lambda_e) + R(\lambda_e) \alpha$, where 
$R(\lambda) = [\, r_1(\lambda) \,\, \cdots \,\, r_{n-1}(\lambda) \,] \in \mathbb{R}^{n \times(n-1)}$ is a matrix polynomial of degree $d$
whose columns are in $\{ H_{h_i} \Adj{P(\lambda)} w \}^{\perp}$, and $\beta \in \mathbb{R}$, $\alpha \in \mathbb{R}^{n-1}$ are the coordinates of $z$ 
in the basis \eqref{eq:orthogonal_basis}.
Then, we have $\projection{\nabla h_i}(\lambda) z = \norm{\nabla h(\lambda)}{}^2 R(\lambda) \alpha$, and the left side of \eqref{eq:lift_curvature_condition} becomes
\begin{align}
z^\mathsf{T} S_{null}(\lambda_e) z \!=\! \norm{\nabla h(\lambda_e)}{}^4 \alpha^\mathsf{T} S(\lambda_e) \alpha
\label{eq:lift_equilibrium_condition2}
\end{align}
with $S(\lambda)$ as defined in \eqref{eq:stability_polynomial}. 
Since $\alpha$ is arbitrary, from \tref{thm:curvature}, we conclude that $\nu_e$ is stable if and only if $S(\lambda)$ is negative semi-definite at $\lambda_e$. 
Otherwise, $\exists \alpha \in \mathbb{R}^{n-1}$ such that \eqref{eq:lift_equilibrium_condition2} is strictly positive, and $\nu_e$ is unstable. 
This proves {(iv)}.

The polynomial matrix $S(\lambda) \in \mathbb{R}^{n-1 \times n-1}$ has maximum degree 3 (odd) and its leading coefficient is positive semi-definite. 
That means that there exist threshold values $\sigma_+$, $\sigma_-$ such that
$S(\lambda) \ge 0$ for all $\lambda \ge \sigma_+$ and 
$S(\lambda) \le 0$ for all $\lambda \le \sigma_-$, respectively.
Furthermore, its determinant $|S(\lambda)|$ has a maximum of $3(n-1)$ real roots, which are exactly the values of $\lambda$ where the eigenvalue curves of $S(\lambda)$ change sign. { By the continuity of $S(\lambda)$}, that means that all $n-1$ eigenvalue curves of $S(\lambda)$ must go from negative to positive as $\lambda$ increases. Then, excluding these $n-1$ roots, a total of $2(n-1)$ roots remain, which can result in a maximum of $n-1$ negative semi-definite intervals for $S(\lambda)$ (in case all roots of $|S(\lambda)|$ are real), plus the first negative semi-definite interval $\mathcal{I}_1 = (-\infty, \sigma_-]$. 
Therefore, in the worst case, $n$ negative semi-definite intervals for $S(\lambda)$ exist. This proves {(v)}. \myqed 


%
\end{pf}
\vspace{-4mm}
{ These properties provide useful geometric information about the Q-function, allowing us to draw conclusions on the number of boundary equilibrium solutions and on their stability properties. 
The next result uses the Q-function properties to provide a necessary and sufficient condition for quadratic CLF compatibility, considering the class of LTI and driftless full-rank systems.}
\begin{cor}
\label{cor:quadratric_compatibility}
Under Assumptions \ref{assumption:initial_state}-\ref{assumption:clf_condition}, { for systems of the classes (i) or (ii), a quadratic CLF \eqref{eq:parametric_CLF} is compatible with a quadratic CBF \eqref{eq:parametric_CBF}} if and only if $S(\lambda_e)$ is not negative semi-definite at the positive real roots of the polynomial $z(\lambda) = n(\lambda) - |P(\lambda)|^2$.
\end{cor}
\vspace{-4mm}
\begin{pf}
This result follows directly from property (iv) of \tref{thm:q_properties}. Notice that the positive roots of the polynomial $z(\lambda) = n(\lambda) - |P(\lambda)|^2$ correspond to the equilibrium solutions $q(\lambda) = \frac{n(\lambda)}{|P(\lambda)|^2} = 1$. If all of them occur at the regions where the stability polynomial matrix $S(\lambda)$ is not negative semi-definite, then these roots correspond to unstable boundary equilibrium points. By \aref{assumption:clf_condition}, no interior equilibrium points other then the CLF minimum exist. Then, we conclude that the CLF \eqref{eq:parametric_CLF} is compatible with the $i$-th CBF. \myqed
\end{pf}
\vspace{-4mm}
\subsection{Numerical Examples}
\begin{figure}[htbp!]
\centering
  \includegraphics[width=1.0\columnwidth]{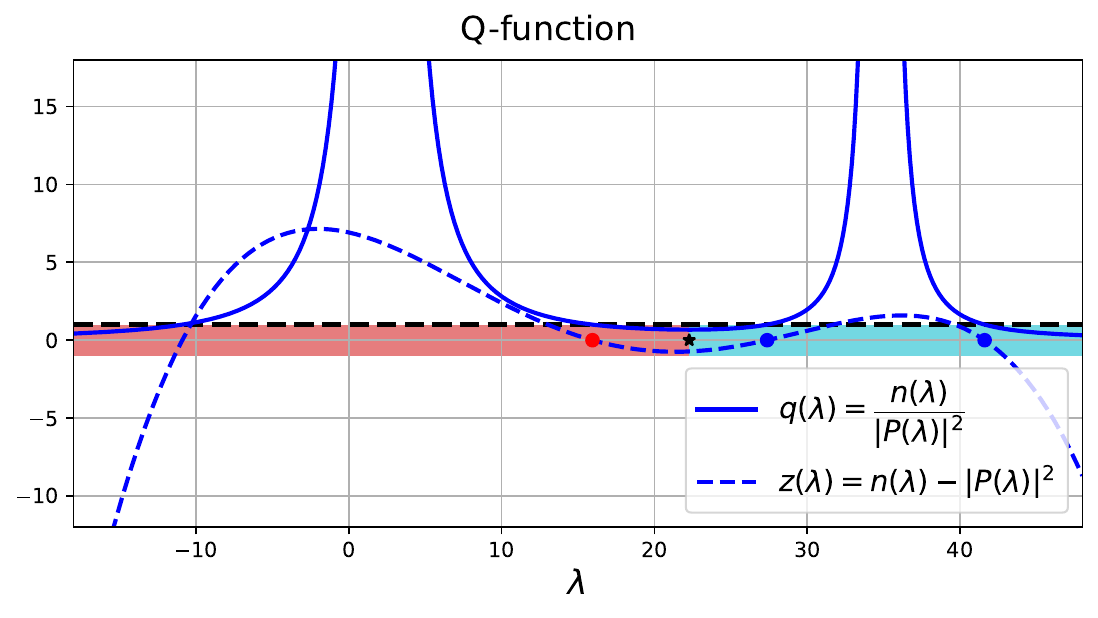}
\vspace{-8mm}
\caption{Q-function and semi-definiteness intervals of $S(\lambda)$, for a two-dimensional LTI system.}
\label{fig:qfunction}
\end{figure}
{\bf Example 1.}
Figure \ref{fig:qfunction} shows the graphs of $q(\lambda)$ and $z(\lambda) = n(\lambda) - |P(\lambda)|^2$ for the 
LTI system $\dot{x}_k = - 2 x_k + u_k$, $k=1,2$, CLF $V(x) = 0.08 x_1^2 + 0.96 x_2^2$, centered at $x_0 = (0,0)$ and CBF $h_3(x) = 0.487 \,(x_1-6)^2 - 0.152 \,(x_1-6) x_2 + 0.069 \, x_2^2 - 0.5$, centered on $c_1 = (6,0)$.
The asymptotes of $q(\lambda)$ occur at the two generalized eigenvalues of the pencil $P$. Notice that $q(0) \ge 1$, or equivalently $z(0) \ge 0$. From \tref{thm:q_properties}{(i)}, this implies \aref{assumption:initial_state}, that is, $x_0 \in \mathcal{C}$, which is indeed true from the CBF expression, since $h(x_0) \ge 0$.
In this example, $S(\lambda)$ is simply a $1\times 1$ matrix polynomial (a scalar polynomial) of degree $3$ with real coefficients. Therefore, $|S(\lambda)|$ has three roots in $\mathbb{C}$. \tref{thm:q_properties}(iv) implies that a maximum of $n = 2$ negative semi-definite intervals can occur for $S(\lambda)$: 
in this example, $S(\lambda) \le 0$ in the interval $(-\infty, 23)$ (red strip) and $S(\lambda) \ge 0$ in the interval $(23,+\infty)$ (blue strip). These intervals are separated by the real root of $|S(\lambda)|$ at $\lambda \approx 23$ (star-shaped dot); the remaining two roots of $|S(\lambda)|$ are complex-conjugates.
There is one solution $z(\lambda_e) = 0$ around $\lambda_e \approx 16$ corresponding to an stable equilibrium point (red dot) 
and other two around $\lambda_e \approx28$ and $\lambda_e \approx 42$ corresponding to unstable equilibrium points (blue dots).
{ From \cref{cor:quadratric_compatibility}, the CLF is not compatible with $h_3$.
}
\begin{figure}[htbp!]
\centering
\includegraphics[width=1.00\columnwidth]{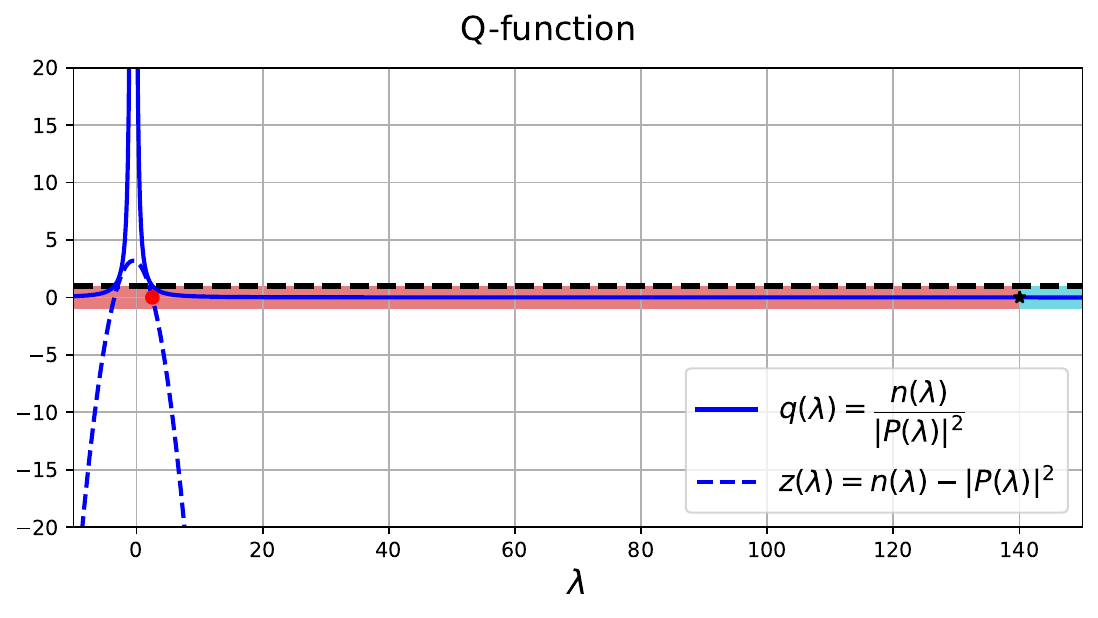}
\vspace{-8mm}
\caption{Q-function and semi-definiteness intervals of $S(\lambda)$, for an underactuated two-dimensional LTI system.}
\label{fig:qfun_underactuated}
\end{figure}

{\bf Example 2.} 
Figure \ref{fig:qfun_underactuated} shows the graphs of $q(\lambda)$ and $z(\lambda) = n(\lambda) - |P(\lambda)|^2$ for the {\it underactuated} two-dimensional LTI system $\dot{x}_1 = x_2$, $\dot{x}_2 = - x_1 - x_2 + u$, $u \in \mathbb{R}$, with the CLF $V(x) = 0.59 x_1^2 - 1.08 x_1 x_2 + 3.91 x_2^2$, centered at $x_0 = (0,0)$ and a similar CBF $h_3(x)$ than in Example 1.
This example is of particular importance since the corresponding control matrix $B$ from $\dot{x} = A x + B u$ is not full-rank.
Many of the observations from Example 1 also hold here, such as $x_0 \in \mathcal{C}$ and $S(\lambda)$ being a scalar polynomial.
However, in this case, $z(\lambda)$ always has two pairs of identical real roots, only one of these with $\lambda \ge 0$ 
$\lambda_e \approx 2.5$, shown in red in \fref{fig:qfun_underactuated}). This means that, in this case, only one boundary equilibrium solution can exist.
Furthermore, the first negative semi-definite interval of $S(\lambda)$ is $(-\infty, 140)$ (red strip).
Moreover, the CLF is not compatible with $h_3$.
%
\begin{thm}[Compatibility Barrier]
\label{thm:compatibility_barrier}
Under Assumptions \ref{assumption:initial_state}-\ref{assumption:clf_condition}, consider the Q-function $q(\lambda)$ associated to the $i$-th CBF. 
Let $\mathcal{I} = (-\infty, \sigma_-]$ be the first negative semi-definite interval of $S(\lambda)$, that is, $S(\lambda) \le 0 \,\, \forall \lambda \le \sigma_-$, and $\epsilon > 1$.
Then, the CLF \eqref{eq:parametric_CLF} is $i$-th compatible if
\begin{align}
B(q) = \min_{ \lambda \in \mathcal{I} \cap \mathbb{R}_{\ge 0} } z_\epsilon(\lambda) \ge 0
\label{eq:compatibility_barrier} \\
S^{\prime}(\lambda) \ge 0 
\label{eq:monotonicity} 
\end{align}
where $z_\epsilon(\lambda) = n(\lambda) - \epsilon |P(\lambda)|^2$ is a polynomial in $\lambda$.
\end{thm}
\vspace{-5mm}
\begin{pf}
In condition \eqref{eq:compatibility_barrier}, $B(q)$ represents a barrier function for the semi-definite interval $\mathcal{I}$: that is, 
under \aref{assumption:initial_state}-\ref{assumption:disjoint_barriers}, if $B(q)$ is non-negative, then there are no roots of 
$z_\epsilon(\lambda)$ in $\mathcal{I} \cap \mathbb{R}_{\ge 0}$.
Then, $q(\lambda) = \frac{n(\lambda)}{|P(\lambda)|^2} \ge \epsilon > 1$ in $\mathcal{I} \cap \mathbb{R}_{\ge 0}$, which means 
that no boundary equilibrium solutions of \eqref{eq:quadratic_equilibrium}-\eqref{eq:quadratic_boundary} exist in $\mathcal{I}$.
%
%
Condition \eqref{eq:monotonicity} ensures that the eigenvalues of $S(\lambda)$ are {\it monotonically increasing}. 
This is a sufficient condition {(but not necessary)} to ensure that no negative semi-definite interval of $S(\lambda)$ other than $\mathcal{I}$ exists , and therefore equilibrium solutions of 
\eqref{eq:quadratic_equilibrium}-\eqref{eq:quadratic_boundary} occurring in $\mathbb{R}_{\ge 0} \setminus \mathcal{I}$ 
correspond to {\it unstable} equilibrium points.
Under \aref{assumption:clf_condition}, no interior equilibrium points other than $x_0$ exist. 
This shows that no stable equilibrium other than $x_0$ exists under conditions \eqref{eq:compatibility_barrier}-\eqref{eq:monotonicity}.
Thus, the CLF is $i$-th compatible. \myqed 
\end{pf}
\vspace{-4mm}
{ In \eqref{eq:compatibility_barrier}, $B(q)$ is a functional of the Q-function $q$ associated to the $i$-th CBF,} being dependent on the system dynamics, CLF and $i$-th CBF geometry. Furthermore, due to condition \eqref{eq:monotonicity}, \tref{thm:compatibility_barrier} is a sufficient, although not necessary condition for CLF $i$-th compatibility. 
\subsection{Compatible CLF Controller}
\label{sec:proposed_controller}

In this section, our objective is twofold: 
(i) to propose a ``compatibilization'' algorithm for computing a compatible CLF from a non-compatible one, and 
(ii) to propose a control strategy to smoothly transform the CLF used in the QP-controller \eqref{eq:QP_control} towards the compatible CLF computed from the compatibilization algorithm, in the regions of the state space where boundary equilibrium points occur.

\begin{defn}
Hessian $H$ is $i$-th compatible if its corresponding CLF $V(x) = \frac{1}{2} \Delta x^\mathsf{T} H \Delta x$
is $i$-th compatible.
\end{defn}


Let $\overline{V}_r(x) = \frac{1}{2} \Delta x^\mathsf{T} H_{\overline{V}_r} \Delta x$ be a reference quadratic CLF
centered on $x_0 \in \mathcal{C}$ (\aref{assumption:initial_state} holds)
and define the following optimization problem related to the $i$-th CBF:
\begin{align}
H_{\overline{V}_i} &= \argmin_{H \in \mathcal{S}^n_+} \,\, \norm{H - H_{\overline{V}_r} }{\mathcal{F}}^2 
\label{eq:compatibilization_opt} \\ 
&\quad \textrm{s.t.} \,\,\, B(q) \ge 0 \tag{compatibility} \\
&\qquad \,\,\, S^{\prime}(\lambda) \ge 0 \tag{monotonicity} \\
&\qquad \,\, H A + A^\mathsf{T} H \le 0 \tag{CLF condition}
\end{align}
The result of optimization \eqref{eq:compatibilization_opt} is the Hessian $H_{\overline{V}_i}$ of the closest quadratic $\overline{V}_i(x) = \frac{1}{2} \Delta x^\mathsf{T} H_{\overline{V}_i} \Delta x$ to the reference CLF $\overline{V}_r$ satisfying:\\
{\bf(i)} For LTI systems, $\overline{V}_i$ is a valid CLF since it satisfies the CLF condition: $L_f \overline{V}_i \le 0 \rightarrow H_{\overline{V}_i} A + A^\mathsf{T} H_{\overline{V}_i} \le 0$. 
For driftless systems: the CLF condition is always satisfied due to $f(x) = A x = 0$ ($A=0$). \\
%
{\bf(ii)} $\overline{V}_i$ is $i$-th compatible (due to \tref{thm:compatibility_barrier}).
Likewise, $H_{\overline{V}_i}$ is $i$-th compatible. \\
%
Thus, if the reference $H_{\overline{V}_r}$ is already $i$-th compatible, the result of \eqref{eq:compatibilization_opt} is simply $H_{\overline{V}_i} = H_{\overline{V}_r}$.
\vspace{-2mm}
\begin{rem}
In optimization \eqref{eq:compatibilization_opt}, the barrier $B(q)$ depends on polynomials such as $n(\lambda)$, $|P(\lambda)|^2$ 
and $S(\lambda)$, which can be efficiently computed using computational methods for polynomial manipulation.
These polynomials depend on the system and CLF-CBF parameters. Therefore, $B(q)$ is dependent on the optimization variable $H$, and must be recomputed at 
each solver iteration. Any solver supporting non-convex constrained optimization could be used, such as Sequential Least Squares Programming (SLSQP) \cite{Bonnans2006}.
\end{rem}

Let $\{ H_{\overline{V}_1}, \cdots, H_{\overline{V}_N} \}$ be a set of $N$ compatible Hessians computed using \eqref{eq:compatibilization_opt}, where 
$H_{\overline{V}_i}$ is the closest $i$-th compatible Hessian to the reference $H_{\overline{V}_r}$. 
%
%
Define a parametric CLF $\overline{V}(x,\pi) = \frac{1}{2} \Delta x^\mathsf{T} H_{\overline{V}}(\pi) \Delta x$ with parametrized Hessian given by
\begin{align}
H_{\overline{V}}(\pi) = L(\pi)^\mathsf{T} L(\pi) \in \mathcal{S}^n_{+}
\label{eq:parameterized_hessian}
\end{align}
where $\pi \in \mathbb{R}^{\dim{\mathcal{S}^n}}$ is a state vector defining the geometry of the level sets of $\overline{V}$.
%
We seek to design a controller for $\pi$, so that the level sets of $\overline{V}$ are dynamically changed.
Using an integrator as the CLF shape state dynamics $\dot{\pi} = u_\pi$, define a Lyapunov function candidate as
\begin{align}
    V_\pi(\pi, H_r) = \frac{1}{2} \norm{ H_{\overline{V}}(\pi) - H_r }{\mathcal{F}}^2
    \label{eq:clf_shape}
\end{align}
%
where $H_r \in \mathcal{S}^n_+$ is a constant Hessian.
If a stabilization controller for $\dot{\pi} = u_\pi$ is designed in such a way that $\dot{V}_\pi = \nabla V_{\pi}^\mathsf{T} u_v \le 0$, then $H_{\overline{V}}(\pi)$ approaches $H_r$. This way, the level sets of $\overline{V}$ are smoothly adapted to match the level sets of a CLF with Hessian $H_r$ 
and center $x_0$.

Now, consider the QP controller \eqref{eq:QP_control} with the CLF $V$ computed from the inverse CLF transformation 
of $\overline{V}(x,\pi)$, as defined in \dref{eq:transCLF}. This transformation is always guaranteed to exist by \pref{prop:properties_transformed_CLF}(ii).
By \tref{theorem:existence_equilibria}, under \aref{assumption:disjoint_barriers}, all of the $i$-th boundary equilibrium points are contained in the set $\mathcal{S}_i$ (defined in \eqref{eq:Si_set}) where the CLF and only the $i$-th CBF constraint are active in the QP \eqref{eq:QP_control}.
If the trajectory of the closed-loop system \eqref{eq:feedback_system} is in the region of attraction of an asymptotically stable 
$i$-th boundary equilibrium point, then the state eventually enters $\Omega^{clf}_i$.
Then, the following strategy is considered: \\
{\bf(i)} if the state is inside $\Omega^{clf}_{i}$, $\overline{V}(x, \pi)$ must converge to $\overline{V}_i$, the closest $i$-th compatible CLF to $\overline{V}_r$, since this will induce a bifurcation on the closed-loop system state space, either removing or rendering the $i$-th boundary equilibrium points unstable.\\
{\bf(ii)} if the state is outside $\cup^N_i \Omega^{clf}_{i}$, $\overline{V}(x, \pi)$ must converge to the reference CLF $\overline{V}_r$. Therefore, if the state is in a deadlock-free region, its dynamics is determined from the reference CLF.

This desired effect can be achieved by the following QP controller for the CLF shape state $\pi \in \mathbb{R}^{\dim{\mathcal{S}^n}}$:
\begin{align}
u_\pi^\star &= 
\argmin_{(u_v,\delta_v)\in\mathbb{R}^{p+1}} \norm{u_\pi}{}^2 + p_{\pi} \delta_v^2
\label{eq:shape_controller} \\
& \textrm{if $x \in \Omega^{clf}_{i}$:} \nonumber \\
& \nabla V_{\pi}( \pi, H_{\overline{V}_i} )^\mathsf{T} u_\pi + \gamma_{\pi} V_\pi(\pi, H_{\overline{V}_i}) \le \delta_v \nonumber \\
& \textrm{otherwise:} \nonumber \\
& \nabla V_{\pi}( \pi, H_{\overline{V}_r} )^\mathsf{T} u_\pi + \gamma_{\pi} V_\pi(\pi, H_{\overline{V}_r}) \le \delta_v \nonumber
\end{align}
where $p_\pi, \gamma_\pi > 0$, set $\Omega^{clf}_i$ as defined in \eqref{eq:Si_set} and $V_{\pi}$ as in \eqref{eq:clf_shape} with reference CLFs drawn from the set $\{ H_{\overline{V}_r}, H_{\overline{V}_1}, \cdots, H_{\overline{V}_N} \}$, depending on the region of the state space where the closed-loop state is located.
With \eqref{eq:shape_controller}, the CLF shape state $\pi$ is controlled to achieve $\overline{V} \rightarrow \overline{V}_i$ when $x \in \Omega^{clf}_{i}$, and $\overline{V} \rightarrow \overline{V}_r$ otherwise.
\begin{rem}
The strategy proposed in \eqref{eq:shape_controller} controls the curvature of the CLF level sets in order to achieve CLF compatibility with respect to the $i$-th active barrier. While this method guarantees that attractive equilibrium points are avoided, the occurrence of other types of attractors such as limit cycles is not theoretically eliminated. { An in depth study of the formation of limit cycles in CBF-based safety-critical control is presented in \cite{mestres2025control}}.
\end{rem}
%
\subsection{Numerical Simulations}
\label{sec:examples}

In this section, we present numerical examples that demonstrate the viability of the proposed method.
The code repository used for producing the results of this section is publicly available at \url{https://github.com/C2SR/CompatibleCLFCBF}.

%
{\bf Simulation 1.} 
Consider again the two-dimensional LTI system $\dot{x}_k = -2 x_k + u_k$, $k=1,2$ from Example 1, whose Q-function was shown in \fref{fig:qfunction} for a given quadratic CLF $\overline{V}_r$ and CBF $h_3$.
This system satisfies \aref{assumption:clf_condition}, and therefore, no interior equilibrium points other than the origin exist.
\begin{figure}[htbp]
\centering
\includegraphics[width=1.0\columnwidth]{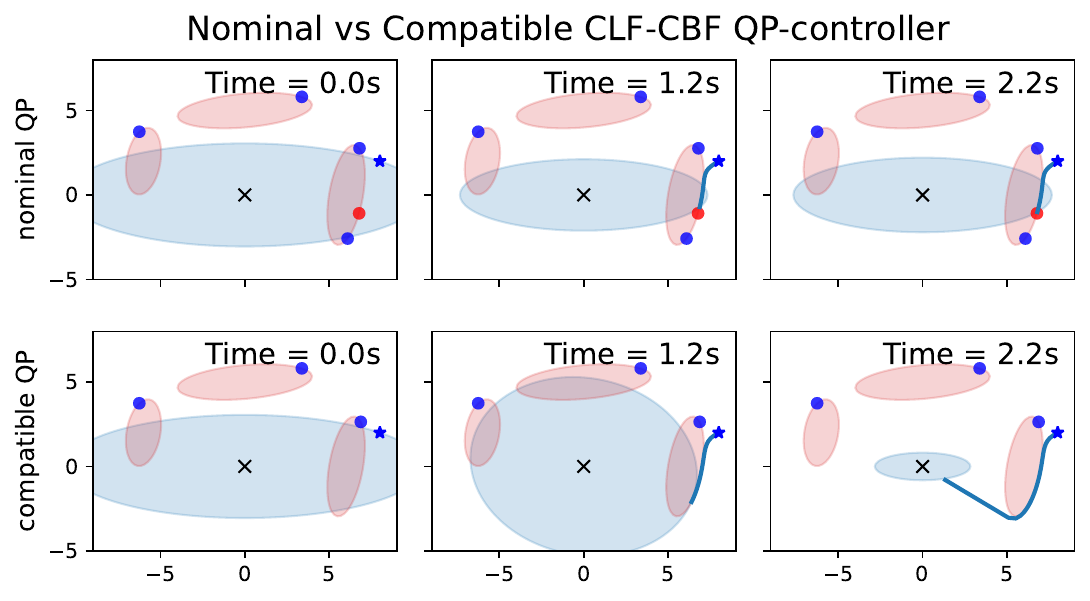}
\vspace{-6mm}
\caption{CLF-CBF controller: fixed CLF vs adaptive strategy. The first row shows the results with a fixed CLF, while the second row shows the results obtained by dynamically adapting the CLF level sets.}
\label{fig:traj_quadratic1}
\end{figure}
%
Here, we consider three quadratic barriers $h_1$, $h_2$, and $h_3$, with unsafe sets shown in \fref{fig:traj_quadratic1} as red ellipses (left, top, and right, respectively). Here, the CLF reference $\overline{V}_r$ and the CBF $h_3$ are the ones used in Example 1 to derive the Q-function from \fref{fig:qfunction}.
The first row of \fref{fig:traj_quadratic1} shows the results obtained using the nominal QP controller \eqref{eq:QP_control} with a fixed reference CLF $\overline{V}_r$ (level set is the blue ellipse) and all three CBF constraints.
As expected, the trajectories converge towards a stable equilibrium point at $\partial \mathcal{C}_3$ (red dot). Two other unstable boundary equilibria also exist at $\partial \mathcal{C}_3$ (blue dots), as seen from the corresponding Q-function in \fref{fig:qfunction}. Each of the remaining CBFs has only one unstable boundary equilibrium. Therefore, $\overline{V}_r$ is compatible with $h_1$ and $h_2$, and is not compatible with $h_3$.
{
Notice that if $\overline{V}_r$ was such that its elliptical level sets were rotated $90$ degrees, their compatibility with the barriers would change: in this case, $\overline{V}_r$ would be compatible with $h_1$ and $h_3$, but not with $h_2$.}

Three compatible CLFs are computed using the optimization \eqref{eq:compatibilization_opt}: $H_{\overline{V}_1}$, $H_{\overline{V}_2}$ and $H_{\overline{V}_3}$, each 
being the $i$-th compatible Hessian closest to the reference $H_{\overline{V}_r}$. 
Here, since $\overline{V}_r$ is already compatible with $h_1$ and $h_2$ (only unstable equilibrium points exist), $H_{\overline{V}_1} = H_{\overline{V}_r}$ and $H_{\overline{V}_2} = H_{\overline{V}_r}$, while $H_{\overline{V}_3}$ is the Hessian of a CLF $\overline{V}_3$ that is compatible with $h_3$. Its level sets are ellipses with slightly smaller eccentricity when compared to $\overline{V}_r$.

The second row of \fref{fig:traj_quadratic1} shows the results obtained using our proposed compatible QP controller \eqref{eq:QP_control} with a CLF $V$
obtained from the inverse transformation of \ref{eq:barV} of a quadratic CLF $\overline{V}(x,\pi)$, parametrized according to \eqref{eq:parameterized_hessian}. 
In this example, we use the simple class $\mathcal{K}$ function $\gamma(V) = \gamma_c V$, where $\gamma_c > 0$ is a constant. Solving \eqref{eq:barV}, the transformed CLF is $\overline{V}(x,\pi) = \frac{1}{2} V^2$, and the inverse transformation is $V = 2 \overline{V}(x,\pi)^{\frac{1}{2}}$.
The Hessian $H_{\overline{V}}(\pi)$ is controlled by our proposed strategy using \eqref{eq:shape_controller}.
From the timestamps, the level sets of $\overline{V}$ dynamically change to match those of $\overline{V}_3$ when $x(t) \in \Omega^{clf}_{3}$, inducing a bifurcation that removes the stable point (and one of the unstable points as well). Only one stable point remains at the boundary $\partial \mathcal{C}_3$.
The system trajectories converge towards the origin for all tested initial conditions, and the level sets of $\overline{V}$ 
converge back to match those of $\overline{V}_r$ after the state leaves $\Omega^{clf}_{3}$, as seen from the second row, third column of \fref{fig:traj_quadratic1}.



{\bf Simulation 2.}
Consider again the {\it underactuated} two-dimensional LTI system $\dot{x}_1 = x_2$, $\dot{x}_2 = - x_1 - x_2 + u$, $u \in \mathbb{R}$
from Example 2. For this system, the state transition matrix $A$ is Hurwitz stable and the pair $(A,B)$ is controllable.
The same CBFs $h_1, h_2$ (left and top) from Simulation 1 were used, and CLF reference $\overline{V}_r(x)$ and CBF $h_3(x)$ (right) were the ones from Example 2, previously used to obtain the Q-function from \fref{fig:qfun_underactuated}.
\begin{figure}[htbp]
\centering
\includegraphics[width=1.0\columnwidth]{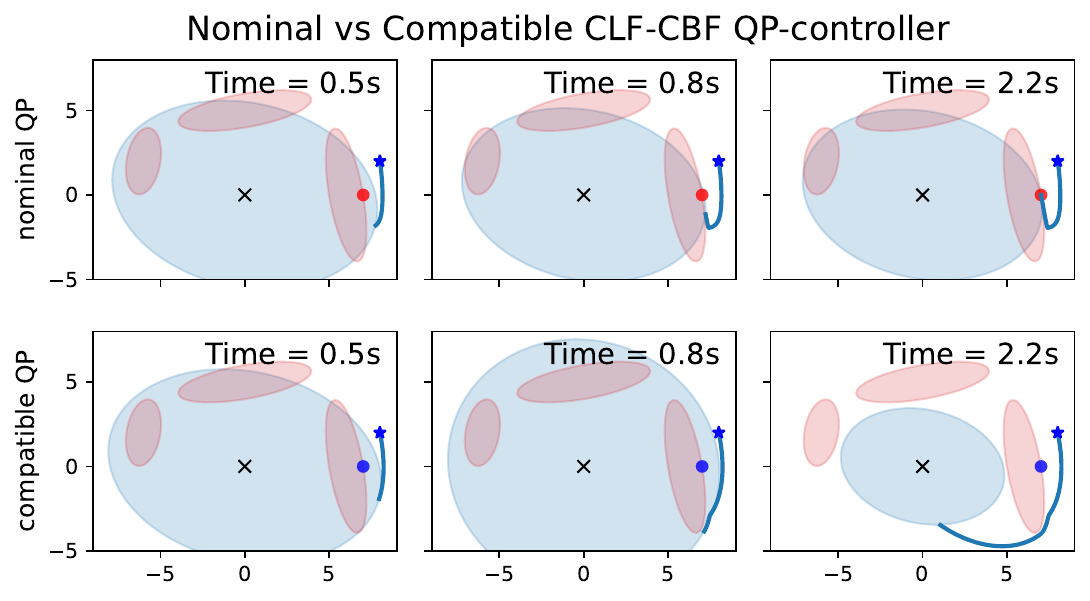}
\vspace{-6mm}
\caption{Compatible QP controller: fixed CLF vs adaptive strategy with underactuated system.}
\label{fig:traj_quadratic2}
\end{figure}

From Example 2, using the nominal QP controller \eqref{eq:QP_control} with $\overline{V}_r$ results in a stable boundary equilibrium point in $\partial \mathcal{C}_3$, as shown by the converging trajectory in the first row of \fref{fig:traj_quadratic2}.
Once again, three compatible Hessians were computed using the optimization \eqref{eq:compatibilization_opt}: $H_{\overline{V}_1}$, $H_{\overline{V}_2}$ and $H_{\overline{V}_3}$, each being the $i$-th compatible Hessian closest to the reference CLF Hessian $H_{\overline{V}_r}$, shown by the level sets shown in blue in the first row of \fref{fig:traj_quadratic2}.
Here, again $H_{\overline{V}_1} = H_{\overline{V}_2} = H_{\overline{V}_r}$ since the reference CLF 
$\overline{V}_r$ is already compatible with barriers $h_1$ and $h_2$ (in this case, no boundary equilibrium points exist at $\partial \mathcal{C}_1$ or $\partial \mathcal{C}_2$). The third computed CLF Hessian $H_{\overline{V}_3}$ is such that the CLF $\overline{V}_3 = \frac{1}{2} x^\mathsf{T} H_{\overline{V}_3} x$ is compatible with $h_3$.

In the second row of \fref{fig:traj_quadratic2}, the results for the QP controller with the adaptive strategy for the CLF level sets are shown. The Hessian $H_{\overline{V}}(\pi)$ is once again controlled by our proposed strategy using \eqref{eq:shape_controller}.
As shown in the second row of \fref{fig:traj_quadratic2}, the level sets of $\overline{V}$ dynamically change to match those of $\overline{V}_3$ when $x(t) \in \Omega^{clf}_{3}$, inducing a bifurcation that effectively transforms the previously stable equilibrium point into an unstable one in $\partial \mathcal{C}_3$.
Therefore, instead of converging towards the boundary $\partial \mathcal{C}_3$, the trajectory circulates the third obstacle and converges to the origin.

\section{Conclusion}


In this work, we have fully characterized the conditions for existence of undesirable equilibrium points arising in the CLF-CBF QP framework and 
their stability properties, considering nonlinear, control-affine systems and multiple CBFs.
In particular, we have shown that the conditions for existence and instability of boundary equilibria depend on \eqref{eq:fi} and its derivatives \eqref{eq:Jfi}. 
%
We have demonstrated that boundary equilibrium points cannot be fully removed in general, and through the concept of CLF compatibility, we propose the possibility of choosing the CLF in such a way that all {\it stable} boundary equilibrium points are removed, for certain classes of systems.
%
%
For driftless full-rank and LTI systems, the formation and stability of boundary equilibrium points with quadratic CLF-CBF pairs can be analyzed using the Q-function \eqref{eq:qfunction}, and in particular the theory of matrix polynomials, as described in \sref{sec:clf_compatibility}.
Additionally, we have proposed an algorithm for computing a compatible quadratic CLF with respect to a quadratic CBF, and a control strategy to modify the CLF in \eqref{eq:QP_control}, aiming to remove all stable equilibrium points from the closed-loop dynamics.
{
Future related research aims towards extending the Q-function theory for CLF compatibility for general classes of systems and CLF-CBF pairs.}

\bibliographystyle{plain}
\bibliography{references}

\end{document}